\documentclass[11pt]{amsart} 
\usepackage{amsfonts}
\usepackage{amssymb} 
\usepackage{graphicx}
\usepackage{color}
\usepackage[all]{xy}
\usepackage{lmodern}
\usepackage{hyperref}
\hypersetup{backref,colorlinks=true,citecolor=blue,linkcolor=blue}

\input epsf

\newcommand{\bea}{\begin{eqnarray}}
\newcommand{\eea}{\end{eqnarray}}
\newcommand{\be}{\begin{equation}}
\newcommand{\ee}{\end{equation}}

\newtheorem{theorem}{Theorem}[section]
\newtheorem{proposition}[theorem]{Proposition}
\newtheorem{lemma}[theorem]{Lemma}
\newtheorem{corollary}[theorem]{Corollary}
\theoremstyle{definition}

\newtheorem{example}[theorem]{Example}

\newtheorem{remark}[theorem]{Remark}

\renewenvironment{proof}{{\noindent\bf Proof.}}{\hfill $\Box$\par\vskip3mm}

\begin{document}

\begin{flushright}
\begin{tabular}{l}
CALT-68-2869 \\

\\ [.3in]
\end{tabular}
\end{flushright}

\title[Topological recursion for chord diagrams]{Topological recursion for chord diagrams,
RNA complexes, and cells in moduli spaces}
\author [J.\ E.\ Andersen, L.\ O.\ Chekhov, R.\ C.\ Penner, C.\ M.\ Reidys, P.\ Su{\l}kowski]
{J{\o}rgen E. Andersen$^{1}$, Leonid O.\ Chekhov$^{3}$,  R.\ C.\ Penner$^{1,2,\star}$,
Christian M. Reidys$^{4}$, Piotr Su{\l}kowski$^{2,5,6}$}

\maketitle

\begin{center}
$^1$ Center for Quantum Geometry of Moduli Spaces, Aarhus University, DK-8000 {\AA}rhus C, Denmark \\
$^2$ Division of Physics, Mathematics and Astronomy, \\ California Institute of Technology, Pasadena, California, 91125 USA \\
$^3$ Department of Theoretical Physics, Steklov Mathematical Institute, Moscow, 119991 Russia \\
$^4$ Department of Mathematics and Computer Science, University of Southern Denmark, DK-5230 Odense M, Denmark \\
$^5$ Institute for Theoretical Physics, University of Amsterdam, \\ Science Park 904, 1090 GL, Amsterdam, The Netherlands \\
$^6$ Faculty of Physics, University of Warsaw, \\ul. Ho{\.z}a 69, 00-681 Warsaw, Poland \\
$^{\star}$ Corresponding author: rpenner@caltech.edu
\end{center}

\bigskip

\begin{abstract} We introduce and study the Hermitian matrix model with potential $V_{s,t}(x)=x^2/2-stx/(1-tx)$, which enumerates
the number of linear chord diagrams of fixed genus with specified numbers of backbones generated by $s$ and chords generated by $t$.
For the one-cut solution, the partition function, correlators and free energies are convergent for small $t$ and all $s$ as
a perturbation of the Gaussian potential, which arises for $st=0$.  This perturbation is computed using the formalism of the topological recursion.  The corresponding enumeration of chord diagrams gives at once the number of RNA
complexes of a given topology as well as the number of cells in Riemann's moduli spaces for bordered surfaces.
The free energies are computed here in principle for all genera and explicitly for genera less than four.
\end{abstract}

\bigskip

\noindent\keywords{Keywords: chord diagrams, matrix model, Riemann's moduli space, RNA, topological recursion}

\eject

\section*{Introduction}
\label{sec:intro}
Consider a collection of $b\geq 1$ pairwise disjoint, oriented and labeled intervals lying in the real line $\mathbb R\subset \mathbb C$, each component of which is called a {\it backbone}.  A {\it chord diagram} $C$ on these backbones is comprised of a collection of $n\geq 0$ semi-circles called {\it chords} lying in the upper half plane whose endpoints lie at distinct interior points of the backbones {\sl so that the resulting diagram is connected}.  

This description of $C$ with its chords in the upper half plane determines a corresponding {\it fatgraph}, namely, a graph in the usual sense of the term together with cyclic orderings on each collection of half-edges incident on a common vertex.  The fatgraph $C$, in turn, determines an associated {\it skinny surface} $\Sigma(C)$ with boundary by replacing each backbone with a 
rectangle to which are attached further semi-circular rectangles respecting orientation, one for each chord of $C$.  In particular, $\Sigma(C)$ contains
$C$ as a  deformation retract; see Figure \ref{fig:surface} for an example, where the chosen under/overcrossings of the rectangles are immaterial to the fatgraph structure.

\begin{figure}[h]
\begin{center}
\includegraphics[width=0.95\textwidth]{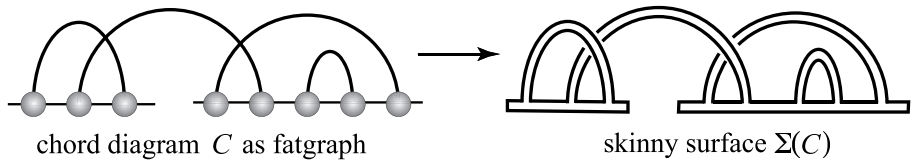}
\end{center}
\caption{From chord diagram to skinny surface.}
\label{fig:surface}
\end{figure}

A chord diagram $C$ therefore naturally determines an oriented surface $\Sigma(C)$, which is connected since $C$ 
is assumed to be connected.  Thus, $\Sigma(C)$ is characterized up to homeomorphism
by its genus $g\geq 0$ and number $r\geq 1$ of boundary components, which we may equivalently associate
to $C$ itself, and the Euler characteristic of $C$ or $\Sigma(C)$ is given by $b-n=2-2g-r$.
This last relation can be used in particular to find the number of boundary components of the surface $\Sigma(C)$
by inspection and then infer the genus. For example in Figure \ref{fig:surface}, we have $b=2$ backbones,
$n=4$ chords and $r=4$ boundary components, which implies that the corresponding surface has genus $g=0$. 

Let $c_{g,b}(n)$ denote the number of isomorphism classes of chord diagrams of genus $g$ with $n$ chords on $b$ labeled backbones whose generating functions 
\be\label{C-intro}
C_{g,b}(z)=\sum _{n\geq 0}c_{g,b}(n)~z^n, ~{\rm for}~g\geq 0,
\ee
 are of our central interest here.  We shall recursively calculate them
using the topological recursion  \cite{ChekhovEynard05,ChEO06,EO07} of a Hermitian one-matrix model 
\be
Z=\int DH~e^{-N{\rm tr}V(H)} = {\rm exp}~{\sum_{g=0}^{\infty} N^{2-2g} F_g},   \label{Zintro}
\ee
where $N$ denotes size of matrices, for a particular potential 
\be
V(x)={{~x^2}\over 2}-{{stx}\over{1-tx}}    \label{Vmatrix-intro}
\ee
generalizing the Gaussian for $st=0$.
The free energy in genus $g$, which we compute, is given by
\be
\label{eq:freeg}
F_g(s,t)= const + \sum_{b\geq 1} {{s^b}\over{b!}}~C_{g,b}(t^2),~{\rm for}~g\geq 0,
\ee
where the constant terms reproduce the Gaussian free energies given by $ \frac{B_{2g}}{2g(2g-2)}$ according to
\cite{abmodel} in each genus $g\geq 2$  with
appropriate modifications for $g=0$ or $1$, and where $B_{2g}$ denote Bernoulli numbers.  
The extra factor $b!$ arises because $C_{g,b}(n)$ counts chord diagrams with labeled backbones as opposed to
unlabeled in the topological recursion, where a permutation of backbones must nevertheless preserve backbone orientations
as we will see.

As has long been known, cf.\  \cite{BIZ,Penner88, Kontsevich92}, Hermitian matrix models are well suited to various computations of combinatorial fatgraph generating functions which are determined by corresponding potentials.
In section \ref{sec:mm}, we derive the specific matrix model potential $V(x)=stx/(1-tx)$ which encodes the solution to our particular combinatorial problem.
This is a two-parameter perturbation of the Gaussian, and the parameters $s$ and $t$ in the expansion of the partition function $Z_{s,t}$ are essentially generating parameters for the numbers of backbones and chords respectively. The entire solution of the model is encoded in the free energies $F_g$ in (\ref{Zintro}), which arise in the large $N$ expansion of the matrix integral.  Provided the 't Hooft parameter $T$ is non-zero (and we shall ultimately take it to be $T=1$ to restrict to the model of the form (\ref{Zintro})), we show there is a unique continuous one-cut extension of the Gaussian which converges for small $t$ and all $s$.  It is this solution that we compute here and show reproduces more elementary enumerative results.

Let us explain in more detail what the topological recursion is and how we can use it to determine the solution of the above matrix model. For a general Hermitian matrix model in its
formal large $N$ expansion, one can write down so-called loop equations, which are Ward identities or Schwinger-Dyson equations for certain multi-linear correlators $W^{(g)}_n(p_1,\ldots,p_n)$ generalizing the resolvent. The leading order equation among those identities specifies a so-called spectral curve, i.e., an algebraic curve on which the resolvent is well defined. It also turns out that all correlators $W^{(g)}_n(p_1,\ldots,p_n)$ and loop equations they satisfy can be encoded entirely in terms of this spectral curve. These loop equations can be solved in a recursive way \cite{ChekhovEynard05,ChEO06,EO07}, and in this manner, free energies $F_g$ (for $g\geq 2$) are completely determined by correlators $W^{(g)}_1(p)$. 

We stress that this entire procedure requires just the knowledge of the spectral curve and not the details of the matrix model from which this curve was derived. An important achievement of Eynard and Orantin \cite{EO07} was to realize that one can use the recursive solution of loop equations to assign correlators $W^{(g)}_n(p_1,\ldots,p_n)$ and $F_g$ to an arbitrary algebraic curve, not necessarily of matrix model origin. On the other hand, it is guaranteed that $F_g$ computed for the spectral curve of a matrix model reproduce the free energies. For various interesting applications of this topological recursion to other problems related to moduli spaces of Riemann surfaces or algebraic curves, cf.\  \cite{BEMS,BKMP,BCMS-FP,BS-mirror,CMS,abmodel,Eynard-moduli,EO07,EO08,MulPen,Norbury,Norbury-Scott,Zhu}.

It was already essentially known that in terms of the cut endpoints $a<b$, the eigenvalue density $\rho(x)$ and the spectral curve $y=M(x)\sqrt{(x-a)(x-b)}$, one can compute the leading free energies $F_0$ and $F_1$ of (\ref{Zintro}), and we follow this procedure for our particular matrix model to derive the multiplicities $c_{g,b}(n)$ of chord diagrams, for $g=0,1$, in appendix \ref{preapp}.  The genus zero free energy can be found from a solution to the variational problem \cite{BIPZ,Marino05}, which leads to the following form of the effective potential
\be
F_0 = -\frac{T}{2} \Big( \int dx \rho(x) V(x) - \lambda t_0 \Big),       \label{F0def}
\ee
where $\lambda$ is a Lagrange multiplier enforcing preservation of the number of eigenvalues $t_0=\int \rho(x)dx$. 
The genus one free energy can be computed as in \cite{ACM1992} (for generalization to more cuts see \cite{chekhov2004}):
\be
F_1 = -\frac{1}{24} \log \big( M(a) M(b) (a-b)^4  \big).    \label{F1def}
\ee
Explicit results for genus zero and one in the model with potential (\ref{Vmatrix-intro}) are presented in appendix \ref{preapp}. These results reproduce known answers (up to three backbones, as given in (\ref{C0-Catalan}), (\ref{Cbb1}) and (\ref{Cg23})), and more importantly provide generating functions for arbitrary numbers of chords and backbones. For example, we determine that generating functions of diagrams  in genus zero and one, on four backbones, with an arbitrary number of chords take the form
\be
\aligned
& C_{0,4} (z)= \frac{24 z^3 (3 + 18 z + 8 z^2)}{(1 - 4 z)^{7}} =  72 z^3 + 2448 z^4 + \ldots  \\
& C_{1,4} (z)= \frac{24 z^5 (715 + 7551z + 12456 z^2 + 2096 z^3) }{(1 - 4 z)^{10}} = 17160 z^5 + \dots  \label{C0414intro}
\endaligned
\ee
To illustrate the combinatorial complexity of chord diagrams and appreciate precisely the generating functions which we derive, we show in appendix \ref{app-count} by explicit enumeration that $c_{0,4}(3)=72$ and $c_{0,4}(4)=2448$ as predicted by (\ref{C0414intro}).

In this paper, apart from the side-note of determining the general $F_0$ (by saddle-point methods) and $F_1$ (due to Chekhov \cite{chekhov2004}) and analyzing them for our model in appendix \ref{preapp},
our main result is a procedure based on topological recursion for calculating any desired one-cut free energy $F_g$, for any $g\geq 2$, on a sufficiently large computer.  We implement it to find, for example,
\be
\aligned
{F}_2 & = -\frac{t^4(1-\sigma)^2}{240 \delta^4 (1 - \delta - 4 \sigma + 3 \sigma^2)^5 (1 + \delta - 4 \sigma + 3 \sigma^2)^5} \times \nonumber \\
& \times \biggl(160 \delta^4 (1 - 3 \sigma)^4 (1 - \sigma)^6 
- 80 \delta^2 (1 - 3 \sigma)^6 (1 - \sigma)^8   \nonumber\\
& + 16 (1 - 3 \sigma)^8 (1 - \sigma)^{10} + \delta^{10} (-16 + 219 \sigma - 462 \sigma^2 + 252 \sigma^3)   \nonumber \\
& + 10 \delta^6 (1 - 3 \sigma)^2 (1 - \sigma)^4 (-16 - 126 \sigma - 423 \sigma^2 + 2286 \sigma^3 - 
     2862 \sigma^4 + 1134 \sigma^5) \nonumber \\
& + 5 \delta^8 (1 - \sigma)^2 (16 + 189 \sigma - 2970 \sigma^2 + 9549 \sigma^3 - 11286 \sigma^4 + 4536 \sigma^5)\biggr ),  \nonumber
\endaligned
\ee
where $\sigma=t(a+b)/2$ and $\delta=t(a-b)/2$.
This can be solved exactly in one of $s,t$ and perturbatively in the other to extract specific $C_{g,b}(z)$ as we will explain.  These methods of computation apply in all higher genera
as well, and our similar but longer expression for $F_3$ is given in (\ref{F2F3}). We again reproduce known results (up to three backbones, given in (\ref{Cbb1}) and (\ref{Cg23})) and then find generating functions for arbitrary numbers of chords and backbones. For instance as  consequences of our computations, we find:
\be
\aligned
C_{2,4}(z) & \, = \, \frac{144 z^7}{(1 - 4 z)^{13}} (38675 + 620648 z + 2087808 z^2 \nonumber \\
    & + 1569328 z^3 +  134208 z^4)  , \nonumber  \\
C_{3,4}(z) & \, = \,  \frac{48z^9}{ (1 - 4 z)^{16}} (53416125 + 1194366915 z + 6557325096 z^2 \nonumber \\
    & + 10738411392 z^3 + 4580024832 z^4 + 236239616 z^5),   \nonumber\\
C_{2,5}(z) & \, = \, \frac{144 z^8}{(1 - 4 z)^{{31}\over 2}} (2543625 + 62424520 z + 375044396 z^2 \nonumber \\
    & + 671666053 z^3 + 314761848 z^4 + 18335696 z^5)  , \nonumber  \\
C_{3,5}(z) & \, = \,  \frac{720z^{10}}{ (1 - 4 z)^{{37}\over 2}} (360380790 + 11275076865 z + 95744892585 z^2  \nonumber \\
    & + 282797424880 z^3 + 291167707410 z^4 + 85497242928 z^5 \nonumber\\
    &+ 3218434848 z^6).   \nonumber
\endaligned
\ee

\bigskip

Let us also comment on the  profound import and wide scope of these results.
The numbers $c_{g,b}(n)$ are of significance in computational biology because they describe the possible non-trivial
topological types of complexes of several interacting RNA molecules as follows.  Each backbone is identified with the
sugar-phosphate backbone (hence the terminology) of a single RNA molecule oriented from its 5' to 3' end.  If two nucleic acids comprising the $b$ RNA molecules participate in a Watson-Crick basepair, then we add a corresponding chord taking care
that chord endpoints in each backbone occur in the correct order corresponding to the primary structure, i.e., the word in 
the four-letter alphabet of nucleic acids that determines the RNA  molecule.
In this way, a complex of interacting RNA molecules determines a chord diagram, and we demand that it is connected in order
to guarantee an appropriate non-triviality of the interaction.  Note in particular that nucleotides not participating
in basepairs play no role in this model, i.e., there are no isolated vertices.  

The genus of the diagram determines the topological complexity of the interaction.  This genus filtration has been profitably employed in a number of studies \cite{OrlandZee02,
Pillsbury05one,Vernizzi05,Bon08,Garg-Deo09,Erba-Zemba09} as well as in  folding algorithms 
\cite{gfold,Pillsbury05two} in the special case of a single RNA molecule (i.e.,  $b=1$), which will be further discussed in section \ref{sec:b=1}.  Likewise, a folding algorithm for two interacting RNA molecules (i.e., $b=2$) pertinent to antisense RNA for instance has been studied in \cite{gfold2}.

At the same time, a simple transform of these numbers $c_{g,b}(n)$ count a sub-class of chord diagrams called ``shapes''.  These are discussed in the next section and give the
number of cells in the ideal cell decomposition \cite{Pennerbook12} of Riemann's moduli space for a surface
of genus $g\geq 0$ with $b\geq 1$ boundary components provided $2g+b>2$.

The numbers computed here using the topological recursion are therefore at once of significance in computational biology and in geometry thus representing a remarkable confluence of biology, mathematics and physics.

This paper is organized as follows.  There is background material in the next section including computations with the Gaussian potential itself for chord diagrams on at most three backbones. In section \ref{sec:mm},
 we introduce our new model and determine the form of its potential (\ref{Vmatrix-intro}). In order to solve this model, we first review some general relevant properties of matrix models in section \ref{sec:matrix} and introduce the formalism of topological recursion in section \ref{sec:recursion}. Finally in section \ref{sec:solution}, we solve the matrix model (\ref{Zintro}) with the potential (\ref{Vmatrix-intro}): we begin by determining its spectral curve and then compute the correlators $W^{(g)}_n(p_1,\ldots,p_n)$ and free energies $F_g$ assigned to it through the topological recursion formalism. As the potential (\ref{Vmatrix-intro}) depends on parameters $s$ (which generates numbers of backbones) and $t$ (which generates numbers of chords), also the spectral curve, and then correlators and free energies depend on these two parameters as well. 
In appendix \ref{preapp}, we present the computations of $F_0$ and $F_1$, which are performed independently of the topological recursion,  and in appendix \ref{app}, we write down the (perturbative) solution of the equations determining cut endpoints in our model.
In the final appendix \ref{app-count}, we
directly enumerate particular classes of chord diagrams  just to confirm
that indeed $c_{0,4}(3)=72$ and $c_{0,4}(4)=2448$.  This enumeration also
gives a sense of the simplest cases of objects we are counting here and can be read first if desired.

\section{Background}

\subsection{Chord diagrams, seeds and shapes}
Two chords $\gamma$ and $\gamma'$ in a chord diagram $C$ with respective endpoints $x,y$ and $x',y'$ are said to be {\it parallel} if $x,x'$, as well as $y,y'$, lie in a common backbone with no chord endpoints in between, where $x< x'$ and $y'<y$.  Parallelism generates an equivalence relation whose equivalence classes are called {\it stacks}.  

Suppose the endpoints $x,y$ of a chord $\gamma$ lie in a common backbone $\beta$.  If there are no chord endpoints between
$x,y$, then $\gamma$ is called a {\it pimple} on $\beta$, and if all chord endpoints in $\beta$ lie between $x$ and $y$, then $\gamma$ is called a  {\it rainbow} on $\beta$.

A {\it seed} is a chord diagram where every stack has cardinality one so that each pimple is a rainbow, and a seed is a {\it shape} provided every backbone has a rainbow.  In particular for a chord diagram of genus zero, an innermost chord with both endpoints on a single backbone is necessarily a pimple, so for example, the only seeds of genus zero on one backbone are the empty diagram with no chords and the diagram with a single chord given by the rainbow, and only the latter is also a shape.

\begin{figure}[h]
\begin{center}
\includegraphics[width=0.95\textwidth]{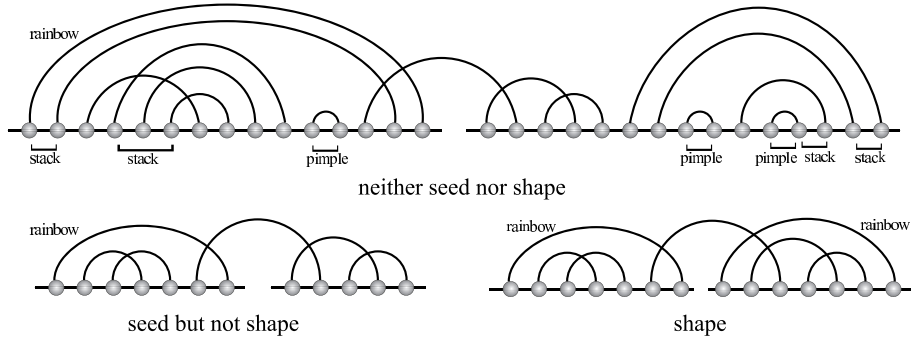}
\end{center}
\caption{Stacks, pimples, rainbows, shapes, seeds.}
\label{fig-items}
\end{figure}

\begin{proposition} \label{prop:fin}
Other than the case $g=0,b=1$ (which was just discussed), a seed of genus $g$ on $b$ backbones with $n$ chords must satisfy 
the inequalities $2g+b-1\leq n \leq 6g-6+5b$, and these constraints are sharp.
\end{proposition}

\begin{proof}
The lower bound on $n$ follows from the fact that $2g+b-1\leq n+1-r$ according to  Euler characteristic considerations together with the 
obvious constraint that $r\geq 1$.  Conversely, a seed which saturates this lower bound necessarily has $r=1$. 

If the skinny surface 
associated to a seed has more than one boundary component, then there must be a chord with different boundary components on its two sides; removing this chord decreases $r$ by exactly one while preserving $g$ again from Euler characteristic considerations.
Define the ``length'' of a boundary component to be the number of chords it traverses counted with multiplicity. If there are $\nu_\ell$ boundary components of length $\ell$, then $2n=\sum_\ell \ell \nu_\ell$ since each side of each chord is traversed exactly once in the boundary.  It follows that $2n=2(b+r+2g-2)\geq \nu_1+2\nu_2+ 3(r-\nu_1-\nu_2)\geq 3r-2b$ since $2\nu_1+\nu_2\leq 2b$
(except in the excluded case for which $2\nu_1+\nu_2\leq 4b$), so
$4(b+g-1)\geq r$. There can thus be at most $4g+4b-5$ such removals of chords to produce a seed with
$r=1$ proving the upper bound on $n$.
\end{proof}

Let $c_{g,b}(n)$, $s_{g,b}(n)$ and $t_{g,b}(n)$, respectively, denote the number of isomorphism classes of (connected) chord diagrams, seeds and shapes of genus $g\geq 0$ with $n\geq 0$ chords on $b\geq 1$ backbones.
In each alphabetic case of $X=C,S,T$, we define corresponding generating functions
\be
\aligned
X_{g,b}(z)&=\sum_{n\geq 0} x_{g,b}(n)~z^{n},   \\
 \label{Xg}
\endaligned
\ee
and let
\be
X_g(z)=X_{g,1}(z).
\ee
Note that whereas $C_{g,b}(z)$ is a formal power series,
$S_{g,b}(z)$ and $T_{g,b}(z)$ are polynomials by Proposition
\ref{prop:fin}.

In particular, $C_0(z)$ is the generating function for the Catalan numbers, i.e., $c_0(n)=c_{0,1}(n)$ is the number of triangulations of a fixed polygon with $n+2$ sides.  They evidently satisfy $c_0(n+1)=\sum_{i=0}^n c_0(i)c_0(n-i)$
with basis $c_0(n)=1$, which gives $C_0(z)=1+z[C_0(z)]^2$, whence
\be
C_0(z)={{1-\sqrt{1-4z}}\over{2z}}.  \label{C0-Catalan}
\ee

\begin{theorem}\label{thm:genfn} The generating functions for seeds and chord diagrams are related by
$$\aligned
C_{g,b}(z)&=[C_0(z)]^b~S_{g,b}\biggl ({{C_0(z)-1}\over{2-C_0(z)}}\biggr ),\cr
S_{g,b}(z)&=\biggl [ {{z+1}\over{1+2z}}\biggr ]^b C_{g,b}\biggl ({{z(1+z)}\over{(1+2z)^2}} \biggr ).\cr
\endaligned$$
Furthermore, the generating functions for seeds and shapes are related by
$$
(1+z)^b~T_{g,b}(z)=z^b~S_{g,b}(z).
$$
\end{theorem}

\begin{proof}  Writing simply $C_0$ for $C_0(z)$ and using $C_0-1=zC_0^2$, we find
$$
{{C_0-1}\over{2-C_0}}={{C_0-1}\over{1-(C_0-1)}}=zC_0^2\sum_{i\geq 0}(zC_0^2)^i=\sum_{j\geq 1}(zC_0^2)^j.
$$
As an argument of $S_{g,b}$, the $j$th term in the sum corresponds to inflating a single arc in a seed to a stack 
of cardinality $j\geq 1$ as well as inserting a genus zero diagram immediately preceding each of the resulting $2j$
chord endpoints.  There is yet another factor $C_0$ arising from the insertion of a genus zero diagram
following the last endpoint of the seed on each backbone accounting for the further factor $C_0^b$.
The resulting chord diagram is connected if and only if the seed itself is connected, and this proves the
first formula.

For the second formula, direct calculation shows that $z={{u(1+u)}\over{(1+2u)^2}}$ inverts the expression
$u={{C_0(z)-1}\over{2-C_0(z)}}={{1-\sqrt{1-4z}}\over{2\sqrt{1-4z}}}$, so the first formula reads
$$S_{g,b}(u)=\biggl [C_0\biggl ( {{u(1+u)}\over{(1+2u)^2}}\biggr ) \biggr ]^{-b}~C_{g,b}\biggl ( {{u(1+u)}\over{(1+2u)^2}}\biggr ).$$
Further direct computation substituting $z={{u(1+u)}\over{(1+2u)^2}}$ into $C_0(z)={{1-\sqrt{1-4z}}\over{2z}}$ shows that
$C_0\bigl ( {{u(1+u)}\over{(1+2u)^2}}\bigr )={{1+2u}\over{1+u}}$, and the expression for $S_{g,b}$ follows.

The third formula is truly elementary since a shape is by definition simply a seed together with a rainbow on each
backbone.
\end{proof}

Seeds are useful insofar as chord diagrams as well as their generalizations allowing isolated vertices which arise in practice can be enumerated
from them, cf.\ \cite{gfold,Reidysbook} for instance. 
Shapes are important because of the following fundamental observation from \cite{APRW,Pennerbook12}:

\begin{theorem}
Provided $2g+b>2$, the set of shapes of genus $g\geq 0$ on $b\geq 1$ backbones with its natural partial order under inclusion
is isomorphic to the ideal cell decomposition of Riemann's moduli space of a surface of genus $g$ with
$b$ boundary components.
\end{theorem}

\begin{proof} Collapsing each backbone to a distinct vertex produces from a shape a fatgraph where each vertex
has a loop connecting consecutive half-edges arising from the rainbow.  The dual fatgraph of this fatgraph arises by interchanging the roles of vertices and boundary components.  The collection of duals of shapes of genus $g$ on $b$ backbones
gives precisely  those fatgraphs that arise in the Penner-Strebel
ideal cell decomposition of the associated Riemann moduli space,
where the collapsed rainbows are dual to the tails discussed in \cite{Pennerbook12}. 
Removal of chords
corresponds to contraction of dual fatgraph edges.
\end{proof}

In light of Theorem \ref{thm:genfn}, the explicit calculation of the generating function for chord diagrams we perform here thus gives as a consequence also the numbers of cells of given dimension in Riemann's moduli spaces of bordered surfaces.

\subsection{The special case of one backbone}\label{sec:b=1}

For one backbone, there is the remarkable formula of Harer-Zagier \cite{HZ}, which we discuss here, that arose in the calculation of the virtual Euler characteristics of Riemann moduli spaces for punctured surfaces and in principle solves for the
various $c_g(n)=c_{g,1}(n)$.  In effect, the Penner matrix model  \cite{Penner88} directly computed the former  without recourse to the latter.  There is a large literature on the Harer-Zagier formula, cf.\ \cite{LZ}  and the references therein including several derivations of it.  Here is this beautiful and striking formula:

\begin{theorem}\label{thm:cg}
We have the identity
$$1+2\sum_{n\geq 0} \sum_{2g\leq n} {{{ c}_g(n)N^{n+1-2g}}\over {(2n-1)!!}}~z^{n+1}=
\biggl ({{1+z}\over {1-z}}\biggr )^N.$$
Furthermore, the $c_g(n)$ satisfy the recursion
$$
(n+1)\,{c}_g(n)  =   2(2n-1)\, {c}_g(n-1)+
                          (2n-1)(n-1)(2n-3)\,{c}_{g-1}(n-2).$$

\end{theorem}

This recursion directly translates into the ODE
\begin{eqnarray}\label{E:ODE}
z(1-4z)\frac{d}{dz}{C}_g(z) +(1-2z){C}_g(z) & = &
\Phi_{g-1}(z),
\end{eqnarray}
where
\begin{eqnarray*}
z^{-2}\Phi_{g-1}(z) =  4z^3\frac{d^3}{dz^3}{C}_{g-1}(z) +
24z^2 \frac{d^2}{dz^2}{C}_{g-1}(z) + 27z\frac{d}{dz}
{C}_{g-1}(z)+3{C}_{g-1}(z)
\end{eqnarray*}
with initial condition ${C}_g(0)=0$. From the ODE  (\ref{E:ODE}) and direct calculation
follow several results:

\begin{corollary}
For any $g\ge 1$, the generating function ${C}_g(z)$ is expressed as
\begin{eqnarray}\label{E:it}
{C}_g(z) = \, \frac{P_g(z)}{(1-4z)^{3g-{1\over 2}}},    \label{Cbb1}
\end{eqnarray}
where $P_g(z)$ is an integral polynomial divisible by      
$z^{2g}$ but no higher power and of degree at 
most $(3g-1)$ with $P_g(1/4)\neq 0$.
\end{corollary}

The first few $P_g(z)$ are given as follows:
\begin{eqnarray*}
P_1(z) &=& z^2,\\
P_2(z) &=& 21z^4\, \left( z+1 \right)\\
P_3(z) &=&  11z^6\, \left( 158\,{z}^{2}+558\,z+135 \right),\\
P_4(z) &=&143z^8\left( 2339\,{z}^{3}+18378\,{z}^{2}+13689\,z+1575 \right),\\
P_5(z) &=&  88179z^{10}\, \left( 1354\,{z}^{4}+18908\,
{z}^{3}+28764\,{z}^{2}+9660\,z+675 \right).
\end{eqnarray*}
It is natural to speculate that the polynomials $P_g(z)$ themselves solve an enumerative problem and 
that there may be a purely combinatorial topological proof of Theorem \ref{thm:cg} based on a construction that in some way
``inflates" these structures using genus zero diagrams, to wit
$$\aligned
{C}_g(z)&= P_g(z) (\sqrt{1-4z})^{1-6g}\\
&=P_g(z)\biggl ( {{{C}_0(z)}\over{2-{C}_0(z)}}\biggr )^{6g-1}\\
&=P_g(z) \bigl ( {C}_0(z)\bigr )^{6g-1}\biggl (1+({C}_0(z)-1 \bigr )+({C}_0(z)-1 \bigr )^2+\cdots \biggr )^{6g-1}\\
&=P_g(z) \bigl ( {C}_0(z)\bigr )^{6g-1}\biggl (1+z({C}_0(z))^2 +z^2({C}_0(z))^4+\cdots \biggr )^{6g-1}.\\
\endaligned$$

Another calculational consequence of the ODE  (\ref{E:ODE}) is an explicit recursion for the numbers of seeds:

\begin{corollary}\label{T:recu1}
We have ${S}_1 (z)= z^2+3\,z^3+3\,z^4+z^5$, and for any $g>1$, the following five-term
recursion for the number of seeds:
\begin{eqnarray}\label{E:recursion1}
(n+1)\,{ s}_g(n) & = &
                     \sum_{h=2}^6 g_h(n)\, { s}_{g-1}(n-h), \quad\text{where}
\end{eqnarray}
\begin{eqnarray*}
g_2(n) & = & (2n-1)(n-1)(2n-3) \\
g_3(n) & = & 2(4n-11)(n-1)(2n-3) \\
g_4(n) & = & (24n^3-180n^2+411n-285) \\
g_5(n) & = & (16n^3-156n^2+464n-414) \\
g_6(n) & = & 4(n-2)(n-4)(n-6).
\end{eqnarray*}
\end{corollary}

The first few seed polynomials are thus given by:
\begin{eqnarray*}
{S}_2 (z)&=& 21\,{z}^{4}+210\,{z}^{5}+840\,{z}^{6}+1785\,{z}^{7}+2205\,{z}^{8}
+ 1596 \,{z}^{9}\\
&&+630\,{z}^{10}+105\,{z}^{11} \nonumber \\
{S}_3 (z) &=& 1485\,{z}^{6}+28413\,{z}^{7}+225379\,{z}^{8}+1006929\,{z}^{9}+2862783
\,{z}^{10}\\
&&+5496645\,{z}^{11} 
 +7325703\,{z}^{12}+6812289\,{z}^{13}+
4348344\,{z}^{14}\\
&&+1819818\,{z}^{15}+450450\,{z}^{16}+50050\,{z}^{17} \nonumber \\
{S}_4 (z)&=& 225225\,{z}^{8}+6687252\,
          {z}^{9}+83809440\,{z}^{10}+605300410\,{z}^{11} \nonumber \\
&& + 2867032740\,{z}^{12}+9542753220\,{z}^{13}+23243924704\,{z}^{14}\\
&& +42438380985\,{z}^{15} +58817592405\,{z}^{16}+62093957640\,{z}^{17}\\
&& + 49660516620\,{z}^{18}+
      29612963952\,{z}^{19} 
 +12768025270\,{z}^{20} \nonumber \\
 &&+ 3763479720\,{z}^{21}+678978300\,{z}^{22}+
 56581525\,{z}^{23} \nonumber \\
{ S}_5 (z)&=& 59520825\,{z}^{10}+2458871415\,{z}^{11}+43395443091\,{z}^{12}+
447114000333\,{z}^{13} \nonumber \\
&&+3067654998408\,{z}^{14}+15065930976096\,{z}^{15}+55383949540920\,{z}^{16}\\
&&
+156912145081692\,{z}^{17}
 +349331909457531\,{z}^{18}+618798832452801\,{z}^{19}\\
 &&+878354703383157\,{z}^{20}+
1001499812704755\,{z}^{21} 
 +915317216039226\,{z}^{22}\\
&&+665807616672198\,{z}^{23}+380384810323518\,{z}^{24}+166997907962886\,{z}^{25}  \nonumber \\
&& +54384464384250\,{z}^{26}+12376076963250\,{z}^{27}  +1756856351250\,{z}^{
28}\\
&&+117123756750\,{z}^{29}  \nonumber \\
\end{eqnarray*}

\begin{remark} \cite{hillary}
With further work starting from this recursion, one can show that the number of seeds is log concave,
i.e., for $2g+1 \leq n \leq 6g-2$, we have \begin{equation}
{ s}_g(n)^2 \geq { s}_g(n+1)\, { s}_g(n-1).\nonumber
\end{equation}
In particular, it follows that the sequence  ${s}_g(2g), { s}_g(2g+1),
\ldots,  { s}_g(6g-1)$ is unimodal.
\end{remark}

 \subsection{The Gaussian resolvents for at most 3 backbones}\label{sec:bleq2}
 
 Provided the number $b$ of backbones is less than four, the loop equations \cite{ACKM,Eynard04}
 for the Gaussian potential permit the computation of the $c_{g,b}(n)$ in terms of $c_{g}(n)$ as we discuss here.  
 This in part explains an important aspect of the computations of the previous section.  
 See also \cite{MS}.
 
 For any $s\geq 1$, the {\it $s$-resolvent}  or {\it $s$-point function} of the potential $V(x)$ is defined to be
\be
\omega_s(x_1,\ldots ,x_s)=N^{s-2}~\bigl <  {\rm tr}~{1\over {x_1-H}}
~\cdots ~  {\rm tr}~{1\over {x_s-H}}\bigr >_{\textrm{conn}},    \label{omega-s}
\ee
where $\bigl < \Phi \bigr >_{conn}$ denotes the part of the correlator
${1\over Z}\int DH~e^{-N{\rm tr}V(H)} \Phi(H)$ arising via Wick's Theorem from
connected diagrams computed as a formal power series using ${\rm tr}~{1\over{x-H}}=\sum_{k\geq 0} {\rm tr}~{{~H^k}\over{x^{k+1}}}$, where the integral is over the $N \times N$ Hermitian matrices with their Haar measure
\be
DH=\biggl (\Lambda_{1\leq i<j\leq N}~d Re H_{ij}\wedge d Im H_{ij}~\biggr )\wedge
\biggl (\Lambda _{i=1}^n d H_{ii}\biggr )   \label{DH}
\ee
and with the partition function $Z=\int DH~e^{-N{\rm tr}V(H)}$.

For a polynomial potential $V(x)=\sum_{k= 1}^d t_kx^k$ for convenience (our methods apply more generally whenever the potential has rational derivative), we find that
\[\aligned
\omega_s(x_1,\ldots ,x_s)
&=\sum_{g\geq 0} N^{-2g}
\sum_{n_1,\ldots ,n_s\geq 0} {1\over{x_1^{n_1+1}\cdots x_s^{n_s+1}}}\\
&\hskip 2cm\sum_{\ell_1,\ell_3,\ldots ,\ell_d\geq 0}\prod _{k=1,k\neq 2}^d(-t_k)^{\ell _k}
\omega^{(g)}_{n_1,\ldots ,n_s}(\{\ell_1,\ell_3,\ldots ,\ell_d\})\\\\
&=\sum_{g\geq 0} N^{-2g}~\omega_s^{(g)}(x_1,\ldots ,x_s),\\
\endaligned     \label{omega-gn}
\]
where $\omega^{(g)}_{n_1,\ldots ,n_s}(\{\ell_1,\ell_3,\ldots ,\ell_d\})$ denotes
the number of isomorphism classes of connected chord diagrams of genus $g$
on $s$ labeled oriented backbones, with respective numbers $n_1,\ldots ,n_s$ of incident half-chords, plus $\ell_k$ unlabeled oriented backbones with $k$ incident half-chords, for $k\geq 1$.

Now specializing to the Gaussian potential $V(x)={1\over 2}x^2$ and letting $\omega_s(x)=\omega_s(x,\ldots ,x)$ for convenience,
we find
$$\omega_s(x)=\sum_{g\geq 0} N^{-2g} ~\omega_s^{(g)}(x)=x^{-s}~\sum_{g\geq 0} N^{-2g}~C_{g,s}(x^{-2})$$
in the earlier notation since $C_{g,b}(z^{-2})=z^s \omega_b^{(g)}(z)$.  The first Schwinger-Dyson equation or Ward identity with equal arguments in this case reads
$\omega_1^2(x)-x \omega_1(x)+1+N^{-2}~\omega_2(x)=0$, from which we solve for the 2-point function with equal arguments in terms of the 1-point function:
\begin{equation}\label{eq:star}
\omega_2(x)=N^2\bigl \{x \omega_1(x)-\omega_1^2(x)-1\bigr \} .
\end{equation}
Likewise, 
$
x \omega_2(x)=2\omega_1(x)\omega_2(x)+{1\over 2} \omega_1''(x)+N^{-2}\omega_3(x)
$
gives the 3-point function with equal arguments
\begin{equation}\label{eq:star2}
\omega_3(x)=N^2\bigl\{ x \omega_2(x)-{1\over 2} \omega_1''(x) -2\omega_1(x)\omega_2(x)\bigr\}
\end{equation}
in terms of $\omega_1(x)$ and its derivatives.

\begin{remark}
One cannot continue this recursion to express higher resolvents $\omega_s(x)$, for $s\geq 4$,
with equal arguments in terms of $\omega_1(x)$ alone because higher Schwinger-Dyson
equations include products of traces (for example,
$\omega_4(x)$ involves $\bigl <{\rm tr}~{1\over{(x-H)^3}}~{\rm tr}~{1\over{x-H}} \bigr>_{conn}$)
which cannot be evaluated in this manner.  We wonder if this is perhaps related to the behavior \cite{Penner08} of posets of chord diagrams: in the case of genus zero for $s$ backbones with $s=1,2,3$, the geometric realization of this poset is homeomorphic to a sphere; for $s=4$, it is homeomorphic to a simply connected manifold which is not a sphere; and for $s>4$, it is an increasingly singular non-manifold.
\end{remark}

\begin{theorem} \label{thm:g2g3}                 
For any $g\geq 0$, we have
\be
C_{g,2}(x) = {{P_{g,2}(x)}\over{(1-4x)^{3g+2}}} ~{\it and}~ C_{g,3}(x) = {{P_{g,3}(x)}\over{(1-4x)^{3g+5-{1\over 2}}}},  \label{Cg23}
\ee
where
\[\aligned
P_{g,2}(x)&={{x^{-1}P_{g+1}(x)-\sum_{h=1}^g P_h(x)P_{g+1-h}(x)}}\ and \\
P_{g,3}(x)&=x^{-1}P_{g+1,2}(x)-2\sum_{h=1}^{g+1} P_h(x)P_{g+1-h,2}(x)\\
&-P_{g+1}(x)\biggl [1+(42+60g)x+(96+336g+288g^2)x^2 \biggr ]\\
&-P'_{g+1}(x) (1-4x)\biggl [5x+(20+48g)x^2 \biggr ] \\
&- P''_{g+1}(x)(1-4x)^2 \biggl [2x^2 \biggr ].\\
\endaligned\]
\end{theorem}

\begin{proof}
It follows from (\ref{eq:star}) that $\omega_2(z)$ is given by
$$\sum_{g\geq 0} N^{-2g} \omega_2^{(g)}(z)=N^2\biggl\{
z\sum_{g\geq 0}N^{-2g}\omega_1^{(g)}(z)-\bigl [
\sum_{g\geq 0} N^{-2g} \omega_1^{(g)}(z)\bigr ]^2-1
\biggr \}.
$$
Extracting the coefficient of $N^{-2g}$, we conclude that
$$\omega_2^{(g)}(z)=z\omega_1^{(g+1)}(z)-\sum_{h=0}^{g+1} \omega_1^{(h)}(z)\omega_1^{(g+1-h)}(z),$$
i.e.,
$$z^{-2}C_{g,2}(z^{-2})=z z^{-1}C_{g+1}(z^{-2})-\sum _{h=0}^{g+1} z^{-2}C_h(z^{-2})C_{g+1-h}(z^{-2}).$$
Setting $x=z^{-2}$,
we find
$$\aligned
C_{g,2}(x)&=x^{-1} C_{g+1}(x)-\sum_{h=0}^{g+1} C_h(x)C_{g+1-h}(x)\\
&=x^{-1}C_{g+1}(x)-\sum_{h=1}^gC_h(x)C_{g+1-h}(x)-2C_0(x)C_{g+1}(x),\\
\endaligned$$
and rewriting this using  (\ref{E:it}) gives
$$\aligned
C_{g,2}(x)&=\biggl\{{1\over x}-{{1-\sqrt{1-4x}}\over{x}} \biggr \}
{{P_{g+1}(x)}\over{(1-4x)^{3g+3-{1\over 2}}}}\\
&\hskip .5cm-{1\over{(1-4x)^{3g+2}}}\sum_{h=1}^g P_h(x) P_{g+1-h}(x)\\
&={{P_{g,2}(x)}\over{(1-4x)^{3g+2}}}\\
\endaligned$$
as was claimed.

Starting from (\ref{eq:star2}), the analogous but slightly more involved 
computation produces
the asserted expression for $C_{g,3}(x)$.
\end{proof}

For later comparison to our one-cut solution, we here record the first several such polynomials on two or three backbones in low genus:
\begin{eqnarray*}
P_{0,2}(x)&=&x,\\
P_{1,2}(x)&=&x^3(20x+21),\\
P_{2,2}(x)&=&x^5\, \left(1696x^2+6096x+1485 \right),\\
P_{3,2}(x)&=&x^7\, \left(330560x^3+2614896x^2+1954116x+225225 \right),\\\\
P_{0,3}(x)&=&2x^2(3+4x),\\
P_{1,3}(x)&=&12 x^4 (45 + 207 x + 68 x^2),\\
P_{2,3}(x)&=&6 x^6 (15015 + 137934 x + 197646 x^2 + 27592 x^3),\\
P_{3,3}(x)&=&8 x^8 (3132675 + 46335375 x + 143262162 x^2 \\
&&+ 98362965 x^3 + 
   7468348 x^4),\\
\end{eqnarray*}


\section{The potential $V(s,t,x)$}\label{sec:mm}

If $C$ is a disjoint union of chord diagrams, then let 
$$\aligned
b(C)~&{\rm denote~its~number~of~backbones},\\ 
n(C)~&{\rm denote~its~number~of~chords},\\ 
r(C)~&{\rm denote~its~number~of~boundary~cycles},\\ 
Aut(C)~&{\rm denote~its~automorphism~group}\\
&{\rm permuting~oriented~backbones},~{\rm and}\\ 
g(C)~&{\rm denote~the~sum~of~genera~of~its~components}.\\
\endaligned$$
Furthermore, set $\mu_N=2^{N/2} \pi^{N^2/2}$ and for any tuple $v_1,v_2,\ldots ,v_K$ of non-negative integers,
define
$$P_{v_1,\ldots ,v_K}(s,t,N)={1\over{\mu_N}}~{1\over{\prod_k v_k!}}\int e^{-{\rm tr}~H^2/2}~\prod_k
\biggl ( s~{\rm tr}~(tH)^k \biggr )^{v_k}~DH$$
for the integral as before over the $N \times N$ Hermitian matrices with respect to the Haar measure $DH$ given in (\ref{DH}).

\begin{lemma} 
For any tuple $v_1,\ldots ,v_K$, parameters $s,t$ and natural number $N$, we have
$$P_{v_1,\ldots ,v_K}(s,t,N)=\sum {{N^{r(C)} ~s^{b(C)}~ t^{2n(C)}}\over{\# Aut(C)}},$$
where the sum is over all isomorphism classes of disjoint unions of chord diagrams with $v_k$ backbones
having $k$ incident half-chords, for $k=1,\ldots ,K$.
\end{lemma}

\begin{proof} The special case $s=t=1$ is entirely analogous to that of Theorem 2.1 of \cite{Penner88},
where the replacement of ${1\over k}~{{\rm tr}~H^k}$ there by ${{\rm tr}~H^k}$ here corresponds to replacing
fatgraphs there by chord diagrams here, i.e., distinguishing one sector at each vertex of the fatgraph
to represent the location of the backbone in the chord diagram kills the cyclic permutations about the vertices.
The general case then follows easily since $b(C)=\sum_k v_k$ and $2n(C)=\sum_k k v_k$.
\end{proof}

A standard computation as on p.\ 49 of \cite{Penner88} thus gives

\begin{theorem}   
For $Z(s,t,N)={1\over{\mu_N}} \int e^{-N~{\rm tr}~V(s,t,H)}~DH$ with potential
\be
V(s,t,H)={1\over 2} H^2 - {{stH}\over{1-tH}}={1\over 2} H^2-s \sum_{k\geq 1} (tH)^k,     \label{Vmatrix}
\ee
we have
$${log}~Z(s,t,N)=\sum {{N^{2-2g(C)} ~s^{b(C)}~ t^{2n(C)}}\over{\# Aut(C)}},$$
where the sum is over all (connected) chord diagrams $C$.
\end{theorem}

An amusing point, which however does not help with computations, is that the derivative of the logarithmic potential
from \cite{Penner88}
$${d\over{dx}}\biggl ( x+ log(1-x) \biggr )=1+{1\over{x-1}}= {x\over{x-1}}$$
agrees with our potential here up to the Gaussian factor for $s=t=1$.  Let us also emphasize that the
partition function in this section enumerates chord diagrams with unlabeled backbones in contrast
to the considerations in section \ref{sec:b=1}.  (Vertices from insertions are labeled
whereas vertices from the potential are not in Wick's Theorem.)


\section{Matrix models and spectral curves: generalities}    \label{sec:matrix}

As we sketched in the introduction, the main idea of this paper is to solve the problem of counting chord diagrams (or RNA complexes or cells in Riemann's moduli space) using matrix model techniques and the topological recursion. We present this general formalism
in this and the next section. We will apply it in section \ref{sec:solution} to find the free energy expansion of the one-cut Hermitian matrix model with the potential given in (\ref{Vmatrix}). 

Let us consider a general Hermitian matrix model 
\be
Z = \int DH~e^{-\frac{1}{\hbar}{\rm tr}V(H)} = ~{\rm exp}~{\sum_{g=0}^{\infty} \hbar^{2g-2} F_g},     \label{Zmatrix}
\ee
with the measure $DH$ introduced in (\ref{DH}). This model is slightly more general than (\ref{Zintro}) as we replaced the explicit size of matrices $N$ in the exponents by $1/\hbar$. In this section, we present how to find a spectral curve of such a model in the large $N$ limit, with the so-called 't Hooft coupling 
\be
T = \hbar N = const   \label{tHooft}
\ee
kept fixed.
In the next section, we will explain how, from the knowledge of the spectral curve, one can use the topological recursion to determine free energies $F_g$. For a matrix model such as (\ref{Zmatrix}), the spectral curve and in consequence the $F_g$ depend on $T$. To solve
the combinatorial enumeration problem in (\ref{Zintro}), we can simply set $T=1$.

The spectral curve arises from the leading order loop equation of a matrix model, and it can be characterized as a curve on which the resolvent is well-defined. The (all-order) resolvent is defined as a correlator
\be
\omega_1(x) = \hbar \langle \textrm{Tr}\frac{1}{x-M}  \rangle = \sum_{g=0}^{\infty} \hbar^{2g} \omega^{(g)}_1(x)  .   \label{omega-def}
\ee
The leading order term in this expansion $\omega^{(0)}_1(x)$ is often referred to simply as the {\it resolvent}.  It is simply related to the density of eigenvalues $\rho(x)$ which becomes continuous in large $N$ limit and has the following asymptotics for large $x$
\be
\omega^{(0)}_1(x) = \frac{1}{t_0} \int \frac{\rho(x')}{x-x'} d x' \ \underset{x\to\infty}{\sim} \  \frac{1}{x},   \label{omega-infty}
\ee
where
\be
t_0 = \int \rho(x) dx.   \label{t0}
\ee
In this large $N$ limit, the eigenvalues are distributed along compact intervals called cuts, along which in particular the above integrals are performed. Moreover, an inverse relation determines the density of eigenvalues as a discontinuity 
\be
\rho(x) = \frac{1}{2\pi i} \Big(\omega^{(0)}_1(x-i\epsilon) - \omega^{(0)}_1(x+i\epsilon) \Big)
\ee
in the resolvent along these cuts.
The consistency of the above relations requires that
\be
\omega^{(0)}_1(x-i\epsilon) + \omega^{(0)}_1(x+i\epsilon) = \frac{V'(x)}{T}.     \label{omega-V}
\ee

The conditions (\ref{omega-infty}) and (\ref{omega-V}) determine the resolvent $\omega^{(0)}_1(x)$ completely, and it is often possible to guess its form by requiring that these conditions are met. Equivalently, one can find the resolvent using the so-called Migdal formula. In case there is a single cut with endpoints denoted $a$ and $b$, this formula reads
\be
\omega^{(0)}_1(x) = \frac{1}{2T} \oint_{\mathcal{C}} \frac{dz}{2\pi i} \frac{V'(z)}{x-z} \frac{\sqrt{(x-a)(x-b)}}{\sqrt{(z-a)(z-b)}},  \label{omega-Migdal}
\ee
where the contour $\mathcal{C}$ encircles the cut. Moreover, expanding the resulting expression around infinity and imposing the condition (\ref{omega-infty}) provides a relation between 't Hooft coupling $T$ and the cut endpoints. This one-cut resolvent will be sufficient for our computations in this paper, however, it is not hard to generalize the above formula to the case with many cuts having endpoints $a_i$ and $b_i$; in this case, one only needs to replace the radicals in numerator and denominator of the integrand respectively by $\prod_i(x-a_i)(x-b_i)$ and $\prod_i(z-a_i)(z-b_i)$.

It is also convenient to encode the information about the resolvent $\omega^{(0)}(x)$ in a new variable
\be
y = \frac{1}{2}V'(x) - T\omega(x) = \frac{T}{2} \big( \omega^{(0)}_1(x-i\epsilon) - \omega^{(0)}_1(x+i\epsilon) \big),  \label{yVomega}
\ee
which turns out to be related to $x$ by a polynomial equation
\be
A(x,y) = 0.
\ee
This equation therefore defines an algebraic curve, which is called a spectral curve of a matrix model. As we will review below in the context of the topological recursion, one in fact needs to provide a parametric form of this equation, i.e., to represent it in terms of two functions $x=x(p)$ and $y=y(p)$ of a parameter $p$ living on the Riemann surface such that the equation $A(x(p),y(p))=0$ is satisfied identically. For the one-cut solution, the functions $x(p)$ and $y(p)$ are rational, while equation (\ref{yVomega}) takes the form
\be
y = M(x) \sqrt{(x-a)(x-b)},   \label{yMsqrt}
\ee
where $M(x)$ is a rational function.

\begin{example}   \label{example-Gauss}
Let us consider the Gaussian model with the potential $V(x)=\frac{1}{2} x^2$. The Migdal formula (\ref{omega-Migdal}) results in 
\bea
\omega^{(0)}_1(x) & = & \frac{1}{2T}\Big(x - \sqrt{(x-a)(x-b)}   \Big) =  \\
& = & \frac{a+b}{4T} + \frac{(a-b)^2}{16Tx} + \mathcal{O}\big( x^{-2}\big),
\eea
and from the first line, it is clear that (\ref{omega-V}) is satisfied. The expansion in the second line together with the condition (\ref{omega-infty}) then imply  that 
\be
a=-2\sqrt{T}, \qquad \qquad b=2\sqrt{T}.  \label{abGaussian}
\ee
Also, introducing $y$ as in (\ref{yVomega}), we find the following equation for the spectral curve
\be
y^2 = \frac{1}{4}(x-a)(x-b).   \label{spectralGaussian}
\ee
Comparing with (\ref{yMsqrt}), we see that $M(x)=\frac{1}{2}$ in this case, and a parametric form of this algebraic curve can be given for example as follows
\be
\left\{\begin{array}{l}  x(p) = \frac{a+b}{2} + \frac{a-b}{4} \Big( p + \frac{1}{p} \Big)  = -\sqrt{T} \Big( p + \frac{1}{p} \Big),  \\
y(p) =   \frac{a-b}{8} \Big( \frac{1}{p} - p \Big) = \frac{\sqrt{T}}{2} \Big( p - \frac{1}{p} \Big).    \end{array} \right.         \label{paramGauss}
\ee
\end{example}


\section{Topological recursion}    \label{sec:recursion}

The topological recursion introduced in \cite{ChekhovEynard05,ChEO06,EO07} can be used to assign to an algebraic curve $A(x,y)=0$ a series of multi-linear differentials $W^{(g)}_n(p_1,\ldots,p_n)$ and functions $F_g$ called free energies. Free energies $F_g$ are symplectic invariants, i.e., they are invariant with respect to transformations which preserve the two-form $dx\wedge dy$ up to a sign. The multi-differentials transform in an appropriate way under symplectic transformations, and they depend on a particular parametrization $x=x(p), y=y(p)$ of the algebraic curve. The differentials $W^{(g)}_n(p_1,\ldots,p_n)$ are defined recursively by equations which, by definition, coincide with loop equations of a Hermitian matrix model \cite{ACKM}; as shown in \cite{Eynard04}, all details of these equations depend only on the form of its underlying algebraic curve $A(x,y)=0$. Thus, applying the topological recursion to the curve which is a spectral curve of some particular Hermitian matrix model such as (\ref{Zmatrix}), reproduces free energies $F_g$ and correlators $W^{(g)}_n (p_1,\ldots,p_n)$. From the matrix model viewpoint, these correlators are defined as 
$$
\Big{\langle} \textrm{tr} \Big( \frac{1}{x(p_1) - M} \Big) \cdots \textrm{tr} \Big( \frac{1}{x(p_k) - M} \Big) \Big{\rangle}_{{\rm conn}}  = \sum_{g=0}^{\infty} \hbar^{2g-2+n} \frac{W^{(g)}_n (p_1,\ldots,p_n)}{dx(p_1) \ldots dx(p_n)},
$$
which generalizes the resolvent introduced in (\ref{omega-def}), and in particular, $W^{(g)}_1(x) = \omega^{(g)}_1(x) dx$. For $\hbar=N^{-1}$, the functional part of $W^{(g)}_n$ also agrees with $\omega_n^{(g)}$ introduced in (\ref{omega-gn}).

The recursion relations for multilinear differentials $W^{(g)}_n$ are defined in terms of the following quantities. Firstly, these relations are expressed in terms of residues around the branch points
\be
dx(p^*_i)=0 \,.
\label{branchpoint}
\ee
For a point $p$ in the neighborhood of a branch point $p^*_i$ there is a unique
conjugate point $\bar{p}$, defined as
\be
x(p) \; = \; x(\bar{p}) \,.
\label{conjugate}
\ee

Two other important ingredients are the differential 1-form
\be
\omega(p) =  \big( y(\bar{p}) - y(p) \big) dx(p)       \label{def-omega} 
\ee
and the Bergman kernel $B(p,q)$.
The Bergman kernel $B(p,q)$ is defined as the unique meromorphic differential two-form with exactly one pole,
which is a double pole at $p=q$ with no residue, and with vanishing integral over $A_I$-cycles $\oint_{A_I} B(p,q)=0$
(in a canonical basis of cycles $(A_I,B^I)$).
For curves of genus zero as arise from the one-cut solution, the Bergman kernel takes the form
\be
B(p,q) \; = \; \frac{dp\,dq}{(p-q)^2} \,.
\label{Bergman}
\ee
A closely related quantity is a 1-form
$$
dE_q(p) \; = \; \frac{1}{2} \int_q^{\bar{q}} B(\xi,p) \,.
$$
defined in a neighborhood of a branch point $q^*_i$ in terms of which we may express the recursion kernel
\be
K(q,p) \; = \; \frac{dE_q(p)}{\omega(q)} \,.
\label{kernel}
\ee

With the above ingredients, we can now present the topological recursion.
This recursion determines higher-degree meromorphic differentials $W^{(g)}_n (p_1, \ldots, p_n)$ from those of lower degree.
The initial data for the recursion are the following one- and two-point correlators of genus zero
\be
W^{(0)}_1(p) = 0, \qquad \qquad 
W^{(0)}_2(p_1,p_2) = B(p_1,p_2),
\ee
with $W^{(g)}_n=0$, for $g<0$, by definition. For two sets of indices $J$ and $N=\{1,\ldots,n\}$, let us now respectively denote $\vec{p}_J = \{p_i  \}_{i\in J}$, and $\vec{p}_N=\{p_1,\ldots,p_n\}$. The multi-linear correlators are then defined recursively as
\be
\aligned
W^{(g)}_{n+1}(p,\vec{p}_N)
= & \sum_{q^*_i}\underset{q\to q^*_i}{\textrm{Res}}  K(q,p)  \Big(  W^{(g-1)}_{n+2}(q,\bar{q},\vec{p}_N)  \,  + \\
& + \sum_{m=0}^g\sum_{J\subset N} W^{(m)}_{|J|+1}(q,\vec{p}_J) W^{(g-m)}_{n-|J|+1}(\bar{q},\vec{p}_{N/J}  \Big),  \label{top-recursion}
\endaligned
\ee
where $\sum_{J\subset N}$ denotes a sum over all subsets $J$ of $N$. In particular, the lowest order correlators determined by the recursion are
$$
W^1_1(p) =  \sum_{q^*_i}\underset{q\to q^*_i}{\textrm{Res}}  K(q,p) W^0_2(q,\bar{q}) 
$$
and then \cite{W03-1,W03-2}
$$
W^0_3(p,p_1,p_2) = \sum_{q^*_i}\underset{q\to q^*_i}{\textrm{Res}}  K(q,p) \Big( W^0_2(q,p_1) W^0_2(\bar{q},p_2) + W^0_2(\bar{q}, p_1) W^0_2(q,p_2)  \Big) .
$$

Finally, we can define the  free energies $F_g$ in genus $g$. For $g\geq 2$, they come from the corresponding $W^{(g)}_1$:
\be
F_g \; = \; \frac{1}{2g-2} \sum_{q^*_i} \underset{q\to q^*_i}{\textrm{Res}} \phi (q) W^{(g)}_1(q) ,     \label{Fg}
\ee
where $\phi (q)=\int^q y(p) dx(p)$. The expressions $F_0$ and $F_1$, given in (\ref{F0def}) and (\ref{F1def}), can be found independently as discussed, e.g., in \cite{chekhov2004,EO07}. 

\begin{example}   \label{example-Guass-toprec}
Let us apply the formalism of the topological recursion to Example \ref{example-Gauss}, i.e., the algebraic curve $y^2=\frac{1}{4}(x-a)(x-b)$ with a parametrization given in (\ref{paramGauss}). We find two branch points $p^*=\pm 1$ and a global expression for the conjugate point (valid around both of these branch points) $\bar{p}=p^{-1}$. As the Gaussian curve represents the one-cut solution on the genus zero spectral curve, the relevant Bergman kernel is simply given by (\ref{Bergman}), and the recursion kernel (\ref{kernel}) takes the form
\be
K(q,p) = \frac{8q^3}{(a-b)^2(q^2-1) (p-q) (q p-1)}.  \label{kernelGaussian}
\ee
Computing the correlators, we find for instance
\be
\aligned
W^{(1)}_1(p)  \, = \, & - \frac{16p^3\, dp}{(a-b)^2 (p^2-1)^4},   \\
W^{(0)}_3(p_1,p_2,p_3)  \, = \, & \frac{8dp_1 dp_2 dp_3}{(a-b)^2}\big( (p_1 + 1)^{-2} (p_2+1)^{-2} (p_3+1)^{-2}  \\
&   - (p_1 - 1)^{-2} (p_2-1)^{-2} (p_3-1)^{-2} \big),  \\
W^{(2)}_1 \, = \, &  - \frac{86016 dp}{(a-b)^6} \frac{p^7+3 p^9+p^{11}}{(p^2 - 1)^{10}} .     \label{Wgn-Gaussian}
\endaligned
\ee
Finally, the genus-two free energy following from the definition (\ref{Fg}) reads
\be
F_2 = -\frac{16}{15(a-b)^4}.
\ee
Substituting $a,b=\mp 2\sqrt{T}$ according to (\ref{abGaussian}) we find
\be
F_2 = -\frac{1}{240 T^2} .  \label{F2Gauss}
\ee
Computing higher order free energies, we find \cite{abmodel} well-known expressions for the Euler characteristics \cite{HZ,Penner88} of moduli spaces of closed Riemann surfaces of genus $g$ generalizing the above result for genus two, namely,
\be
F_g =  \frac{B_{2g}}{2g(2g-2)} \frac{1}{T^{2g-2}}.   \label{FgGauss}
\ee
\end{example}


\section{Solution of the model}    \label{sec:solution}

In this section, we derive a formal large $N$ expansion and find free energies $F_g$ of the matrix model (\ref{Zmatrix})
\be
Z = \int DH~e^{-\frac{1}{\hbar}{\rm tr}V(H)} = ~{\rm exp } ~{\sum_{g=0}^{\infty} \hbar^{2g-2} F_g},    
\ee
with the potential given in (\ref{Vmatrix}), i.e., $V(H)={1\over 2} H^2 - {{stH}\over{1-tH}}$.
The potential in this model is a function of parameter $s$ which generates the number of backbones and $t$ which generates the number of chords. Thus,  the $F_g$ are also functions of these two parameters. Expanding $F_g$ as a series in $s$ and $t$ we obtain multiplicities labeled by the genus, number of backbones and number of chords, which provide the solution of our model. In fact, free energies depend also on 't Hooft coupling $T=\hbar N$, however, this dependence is rather simple, and all the information we seek can be obtained by setting $T=1$. To solve this matrix model, we follow the steps explained in previous sections: we find a parametrization of the spectral curve and then apply the formalism of the topological recursion to derive associated multi-linear correlators $W^{(g)}_n(p_1,\ldots,p_n)$ and free energies $F_g$. Finally, we present explicit expressions for free energies for small values of $g$ and show that they indeed encode the $c_{g,b}(n)$ discussed in the introduction, i.e., the numbers of chord diagrams for various small genera and numbers of backbones.

\subsection{Spectral curve}

We begin by determining a specific spectral curve. Let us note that the potential (\ref{Vmatrix}) is in fact not bounded below, and there is a singularity as $x\to t^{-1}$. Nonetheless, this singularity does not pose a problem in the formal expansion of the free energy as we will see. Also for this reason, we treat the parameters $s$ and $t$ in the potential (\ref{Vmatrix}) as generating parameters for powers of $x$, and so we effectively consider the behavior of eigenvalues only in the vicinity of the perturbed Gaussian minimum even though there are three critical points of the potential, i.e., we are studying the one-cut solution of the model. We can therefore determine the one-cut semiclassical resolvent $\omega^{(0)}_1(x)$ using the Migdal integral (\ref{omega-Migdal}) leading to
\be
\aligned
\omega^{(0)}_1(x) & \, = \,  \frac{x (x t -1)^2 - s t}{2T (x t -1)^2} +  \frac{\sqrt{(x-a)(x-b)}}{2T}  \times   \\
&  \times \Big( \frac{s t^2 \big(4 +(xt^2-3t)(a+b)-2xt+2ab t^2  \big)}{2(xt-1)^2 \big((at-1)(bt-1) \big)^{3/2}} - 1 \Big),  \label{omegaFull} 
\endaligned
\ee
where $a$ and $b$ denote cut endpoints. Expanding this expression for large $x$ we find that it has a form $\omega^{(0)}_1(x)=c_1 + c_2 x^{-1} + \mathcal{O}(x^{-2})$. However, imposing the asymptotics given in (\ref{omega-infty}) requires that $c_1=0$ and $c_2=1$. These two equations take the following form
\be
\left\{\begin{array}{l} 0 = a + b + \frac{st(at+bt-2)}{\big((at-1)(bt-1) \big)^{3/2}},      \\
16T = (a-b)^2 + \frac{4s \big((2-\frac{(a+b)t}{2})(at+bt-2) + 2abt^2 - 3t(a+b) + 4  \big)}{\big((at-1)(bt-1) \big)^{3/2}}.   \end{array} \right.   \label{ab-equations}
\ee
In principle, one could now solve (\ref{ab-equations}) for $a$ and $b$ and substitute these values back into (\ref{omegaFull}) to get the exact formula for the resolvent, but we know of no closed-form solution.  However, one can determine perturbative expansion of $a$ and $b$ in $s$ and $t$,  and such a representation is sufficient at least as far as perturbative expansion of free energies is concerned. For example, we find that the perturbative expansions in $s$ (with exact dependence on $t$) start as
\be
\left\{\begin{array}{l}  a = a(s,t) = -2 T^{1/2} + s t \frac{(1 - 4t^2 T)^{-1/2}}{1 + 2tT^{1/2}}  + \ldots ,\\
b = b(s,t) = 2 T^{1/2} 
+ s t \frac{(1 - 4t^2 T)^{-1/2}}{1 - 2tT^{1/2}} 
+ \ldots   \end{array} \right.   \label{ab-pert}
\ee
and an analogous expansion in $t$ (with exact dependence on $s$) can be determined.

It is convenient to  introduce half the sum and difference of the cut endpoints $a,b$
\be
S = \frac{a+b}{2},\qquad\quad D = \frac{a-b}{2}.   \label{abSD}
\ee
Taking the square of the first equation in (\ref{ab-equations}) as well as determining the denominator $\big((at-1)(bt-1) \big)^{-3/2}$ from both equations in (\ref{ab-equations}) leads to the following pair of equations
\be
\left\{\begin{array}{l}    4s^2 t^4 (3D^2 - 4T)^6= D^2 (D^2 - 4T)^2 \big(4D^2 - t^2(3D^2 - 4T)^2  \big)^3,     \\
S = \frac{D^2 - 4 T}{t (3 D^2 - 4 T)}. \end{array} \right.   \label{SD-equations}
\ee
The first equation is a polynomial equation for $D$ which can be solved perturbatively to any accuracy in $s$ or $t$, and then the second equation immediately determines $S$. In appendix \ref{app}, we present a perturbative expansion of $S$ and $D$ in $s$ with exact dependence on $t$; from this one can easily find higher order terms in the expansion (\ref{ab-pert}). 

\begin{theorem} \label{thm:conv}
Provided the 't Hooft parameter $T\neq 0$, there is a unique continuous extension of the Gaussian to a solution of (\ref{SD-equations}).  In fact, the extension is smooth in $t$
and converges for all $s$ and small $t$, and it furthermore satisfies $a(s,-t)=-b(s,t)$.
\end{theorem}

\begin{proof}
Consider the polynomial $q_{s,t}(S,D)$ given by the difference of the left- from right-hand side of the first equation (\ref{SD-equations}).  Direct computation shows that
$ \left.\frac{\partial q_{s,t}(S,D)}{\partial D}\right|_{t=0}=512~D^{15}(8T-3D^4)(4T-D^4)$.  Provided this is non-vanishing, the Implicit Function Theorem applies to provide the unique smooth and convergent one-cut solution.  Meanwhile for the Gaussian with $t=0$, we have cut endpoints $\pm 2\sqrt{T}$, so $S=0$ and $D=2\sqrt{T}$.  Thus, $T\neq 0$ implies $D\neq 0$, and $D^4=16T$ further implies that the factors $(8t-3D^4)=-8T$ and $(4T-D^4)=-12T$ are also non-zero as required.  

The asserted symmetry $a(s,t)=-b(s,-t)$ is equivalent to the
statement that $S$ is an odd and $D$ an even function of $t$.  Meanwhile, the first equation in (\ref{SD-equations}) involves only even powers of $t$ so the 
smooth solution from the Implicit Function Theorem must be an even function of $t$.  Insofar as $T$ is independent of $t$, it is even as well, so the second equation
in (\ref{SD-equations}) displays $S$ as an odd function of $t$.
\end{proof}

In what follows, we will work with general expressions for the resolvent and ultimately derive exact expressions for $F_g$ in terms of $S$ and $D$. Substituting the above expansions, we will then find perturbative expansions of $F_g$ in $s$ and $t$. Let us therefore proceed to find the spectral curve in terms of parameters $a$ and $b$ or equivalently $S$ and $D$. This curve is to be expressed in terms of variables $x$ and $y$ defined as in (\ref{yVomega}). In particular, only expressions involving square roots $\sqrt{(x-a)(x-b)}$, given in the second line of (\ref{omegaFull}), will enter the formula for $y$. Moreover, note that the  second line of (\ref{omegaFull}) involves one non-rational expression $\big((at - 1)(bt - 1)\big)^{3/2}$, which however can be replaced by a rational function using the first equation in (\ref{ab-equations}). In consequence, the equation for the spectral curve takes form
\be
y = \frac{1}{2}\sqrt{(x-a)(x-b)}\frac{\big(tx - 1 + \frac{(a+b)t}{4}\big)^2 + \gamma}{(tx-1)^2},      \label{spectral-xy}
\ee
so that comparing with (\ref{yMsqrt}) we find
\be
M(x) = \frac{\big(tx - 1 + \frac{(a+b)t}{4}\big)^2 + \gamma}{2 (tx-1)^2}.   \label{M-new}
\ee
The entire dependence in (\ref{spectral-xy}) on $x$ and $y$ is given explicitly, and $\gamma$ is a rational function of $a$, $b$ and $t$ given below. If we square this equation we obtain an algebraic equation for the spectral curve. In fact, notice that all factors of $t$ in this equation multiply $a$ and $b$, and if we rescale also $x$ and $y$ by $t$, then the entire dependence on $t$ can be encoded in rescaled variables
\be
\xi= x t, \quad \eta = y t,\quad \alpha=a t, \quad \beta=b t, \quad \sigma=S t, \quad \delta = Dt.  \label{rescale-t}
\ee

In particular in terms of variables introduced above, the expression for $\gamma$ takes the form
\bea
\gamma & = & - \frac{(\alpha + \beta) \big( \alpha^2 + \beta^2 + 14(\alpha+\beta-\alpha\beta) -16 \big)}{16 (\alpha + \beta - 2)} \nonumber \\
& = &  \frac{\sigma (4 - 4 \delta^2 - 7 \sigma + 3 \sigma^2)}{4 (\sigma-1)}  .
\eea
Finally in these rescaled variables, the remaining equation reads
\be
(\xi - 1)^4 \eta^2 = \frac{1}{4} (\xi - \alpha)(\xi-\beta) \Big[ \big(\xi-1+\frac{\alpha+\beta}{4})^2 + \gamma \Big]^2.  \label{SpectralCurve}
\ee

This is our desired expression for the spectral curve.  In order to apply the topological recursion, we must further find a parametric form of this equation. 
To this end, note that the square root of the above equation, or equivalently the equation (\ref{spectral-xy}), is a modification of the Gaussian curve $y=\frac{1}{2}\sqrt{(x-a)(x-b)}$ given in (\ref{spectralGaussian}) by a rational function in $x$. The essence of the parametrization of the Gaussian curve  given in (\ref{paramGauss}) was to take care of this square root and devise a function $x=x(p)$ such that $\frac{1}{4}(x-a)(x-b)$ is a complete square of a rational function in $p$, which is then identified with $y$. We can therefore borrow this aspect of the Gaussian parametrization (\ref{paramGauss}) and just modify the form of $y$ by multiplying it by this additional rational function of $x$ in the present case. This leads to the following parametric form of the spectral curve
\be
\left\{\begin{array}{l}  \xi(p) = \frac{\alpha+\beta}{2} + \frac{\alpha-\beta}{4} \Big( p + \frac{1}{p} \Big) = \sigma + \frac{\delta}{2} \Big( p + \frac{1}{p} \Big),    \\
\eta(p) =   \frac{\alpha-\beta}{8} \Big(\frac{1}{p} - p \Big)  \frac{\big(\xi-1+\frac{\alpha+\beta}{4}\big)^2 + \gamma}{(\xi-1)^2}
= \frac{\delta}{4} \Big(\frac{1}{p} - p \Big)  \frac{\big(\xi-1+\frac{\sigma}{2}\big)^2 + \gamma}{(\xi-1)^2}.   \end{array} \right.          \label{parametrize}
\ee
With this parametrization of the spectral curve, we are prepared to compute the free energies using the topological recursion.

\subsection{Free energies}

We can now determine free energies $F_g$, for genus $g\geq 2$, of the spectral curve given in (\ref{SpectralCurve}) as parametrized by (\ref{parametrize}). This curve has genus zero, so the Bergman kernel takes the form (\ref{Bergman}). As the parametrization of the $\xi$ variable is the same as in the Gaussian case (\ref{paramGauss}), we find the same branch points $p=\pm 1$ and globally conjugate points $\bar{p}=p^{-1}$ as given in Example \ref{example-Gauss}. Nonetheless, all other quantities which depend also on $\eta$ have of course a much more complicated form than in the Gaussian case. In particular, the recursion kernel takes the form
\be
\aligned
K(q,p) = & \frac{2q^3 }{(q^2 - 1) (p-q) (q p-1) \delta^2  }  \, 
   \frac{\big((q^{-1} + q) \delta + 
   2 \sigma - 2\big)^2}{4\gamma +  \big((q^{-1} + q) \delta + 
   3 \sigma - 2\big)^2 }    . 
\endaligned   \label{kernel-new}
\ee
With these ingredients, we can compute multi-linear correlators $W^{(g)}_n$. They are given by rather complicated and not very enlightening formulas, which are too lengthy to write down here. However, just to present the simplest example, let us give the formula for $W^{(0)}_3$:
\be
\aligned
W^{(0)}_3(p_1,p_2,p_3) = & \, 2dp_1 dp_2 dp_3 \, \frac{\sigma - 1}{\sigma^2}  \, \times \\
& \Big( \frac{\sigma-\delta-1}{(p_1 + 1)^{2} (p_2+1)^{2} (p_3+1)^{2} 
(1+\delta-4 \sigma+3 \sigma^2) }  \\
&  - \frac{\sigma+\delta-1}{(p_1 - 1)^{2} (p_2-1)^{2} (p_3-1)^{2}  (1-\delta-4 \sigma+3 \sigma^2) }\Big)  .
\endaligned   \label{W03-new}
\ee

At this stage viz.\ Theorem \ref{thm:conv}, we can also confirm that all our results are consistent with the Gaussian model  discussed in Examples \ref{example-Gauss} and \ref{example-Guass-toprec}. The potential (\ref{Vmatrix}) reduces to the Gaussian potential when $st=0$. In both cases
$s=0$ and $t=0$, the system (\ref{ab-equations}) simplifies so that the first equation implies $a=-b$, and then the second equation leads to the result (\ref{abGaussian}) also in agreement with the expansion (\ref{ab-pert}). For $a=-b$, we also find that $\gamma=0$, and the curve (\ref{SpectralCurve}) reduces to the Gaussian curve (\ref{spectralGaussian}). Furthermore, the recursion kernel (\ref{kernel-new}) and all correlators and
(\ref{W03-new}) in particular reduce respectively to (\ref{kernelGaussian}) and (\ref{Wgn-Gaussian}).

Finally with Mathematica \cite{mathematica}, we compute genus two and genus three free energies. As these are the free energies of the curve (\ref{parametrize}), which involves variables $x$ and $y$ rescaled by $t$, let us denote them by $\widetilde{F}_2$ and $\widetilde{F}_3$. We find they are given by the following exact expressions
\be
\aligned
\widetilde{F}_2 & = -\frac{(1-\sigma)^2}{240 \delta^4 (1 - \delta - 4 \sigma + 3 \sigma^2)^5 (1 + \delta - 4 \sigma + 3 \sigma^2)^5} \times \nonumber \\
& \times \Big(160 \delta^4 (1 - 3 \sigma)^4 (1 - \sigma)^6 
- 80 \delta^2 (1 - 3 \sigma)^6 (1 - \sigma)^8   \nonumber\\
& + 16 (1 - 3 \sigma)^8 (1 - \sigma)^{10} + \delta^{10} (-16 + 219 \sigma - 462 \sigma^2 + 252 \sigma^3)   \nonumber \\
& + 10 \delta^6 (1 - 3 \sigma)^2 (1 - \sigma)^4 (-16 - 126 \sigma - 423 \sigma^2 + 2286 \sigma^3 - 
     2862 \sigma^4 + 1134 \sigma^5) \nonumber \\
& + 5 \delta^8 (1 - \sigma)^2 (16 + 189 \sigma - 2970 \sigma^2 + 9549 \sigma^3 - 11286 \sigma^4 + 4536 \sigma^5)\Big)  \nonumber
\endaligned
\ee
\be
\aligned
\widetilde{F}_3 & = \frac{(1-\sigma)^4}{2016 \delta^8 (1 - \delta - 4 \sigma + 3 \sigma^2)^{10} (1 + \delta - 4 \sigma + 3 \sigma^2)^{10}} \times \nonumber \\
& \times \Big(107520 \delta^8 (1 - 3 \sigma)^8 (1 - \sigma)^{12} - 
 61440 \delta^6 (1 - 3 \sigma)^{10} (1 - \sigma)^{14}  \nonumber \\
& + 23040 \delta^4 (1 - 3 \sigma)^{12} (1 - \sigma)^{16} - 
 5120 \delta^2 (1 - 3 \sigma)^{14} (1 - \sigma)^{18} + 
 512 (1 - 3 \sigma)^{16} (-1 + \sigma)^{20} \nonumber \\
 & + \delta^{20} (512 - 
    48603 \sigma + 471708 \sigma^2 - 1645251 \sigma^3 + 
    2623320 \sigma^4 - 1957284 \sigma^5 + 
    555660 \sigma^6)   \nonumber \\
& + \delta^{18} (1 - \sigma)^2 (-5120 - 
    639867 \sigma + 20457090 \sigma^2 - 177935466 \sigma^3 + 
    697802346 \sigma^4 \nonumber \\
&   - 1439492067 \sigma^5 + 
    1623708828 \sigma^6 - 948940272 \sigma^7 + 
    225042300 \sigma^8) \nonumber \\
& + 252 \delta^{10} (1 - 3 \sigma)^6 (1 - \sigma)^{10} (-512 + 
    2970 \sigma + 70785 \sigma^2 - 3780 \sigma^3 - 
    1030050 \sigma^4 \nonumber \\
& + 2252880 \sigma^5 - 1077300 \sigma^6 - 
    1682100 \sigma^7 + 2417850 \sigma^8 - 1129950 \sigma^9 + 
    178605 \sigma^{10}) \nonumber \\
& + 84 \delta^{12} (1 - 3 \sigma)^4 (1 - \sigma)^8 (1280 - 
    17523 \sigma + 256203 \sigma^2 + 5113206 \sigma^3 - 
    22143348 \sigma^4 + 2780703 \sigma^5 \nonumber \\
& + 124222086 \sigma^6 - 
    273630879 \sigma^7 + 265968360 \sigma^8 - 
    126660105 \sigma^9 + 24111675 \sigma^{10}) \nonumber \\
& + 9 \delta^{16} (1 - \sigma)^4 (2560 + 156387 \sigma + 
    605388 \sigma^2 - 68079487 \sigma^3 + 696887576 \sigma^4 - 
    3203644311 \sigma^5 \nonumber \\
& + 8142105636 \sigma^6 - 
    12207659661 \sigma^7 + 10771521264 \sigma^8 - 
    5182076088 \sigma^9 + 1050197400 \sigma^{10}) \nonumber \\
& + 3 \delta^{14} (1 - 3 \sigma)^2 (1 - \sigma)^6 (-20480 - 
    6699 \sigma - 23380056 \sigma^2 + 166860351 \sigma^3 \nonumber \\
& + 520822512 \sigma^4 - 6709872897 \sigma^5 + 22111131168 \sigma^6 - 36421137675 \sigma^7 \nonumber \\
& + 33054947856 \sigma^8 - 15849836352 \sigma^9 + 
    3150592200 \sigma^{10})   \Big)  \nonumber
\endaligned
\ee
The free energies $F_g$ in our model are closely related to these results. Indeed, the original curve is written in terms of $x$ and $y$ while the above free energies are computed for the curve in rescaled variables (\ref{rescale-t}). In general, rescaling a variable $x$ or $y$ by $t$ results in rescaling of the genus $g$ free energy by $t^{2g-2}$. Because we rescale both $x$ and $y$, this means that the free energies are rescaled by $t^{2(2g-2)}$. The genus two and three free energies, with dependence on $t$ reintroduced, are thus given as
\be
F_2 = t^4 \widetilde{F}_2, \qquad \quad  F_3 = t^8 \widetilde{F}_3.  \label{F2F3}
\ee
Note that the entire dependence on $t$ arises upon substituting $\sigma$ and $\delta$ according to (\ref{rescale-t}). 

The expressions (\ref{F2F3}) are our desired results, and all we need do is substitute the expressions for $\sigma$ and $\delta$ or in fact $a$ and $b$ given in principle by solving the equations (\ref{ab-equations}). In practice however, what we can do is to substitute the series expansion of $a$ and $b$ in terms of $s$ and $t$ given in (\ref{ab-pert}) or to better accuracy in appendix \ref{app}. This finally leads to the ultimate expansion of $F_2$ and $F_3$ in  $s$ and $t$ which we seek. Let us present results after setting 
\be
T=1 \ \textrm{and} \ z=t^2.   \label{T1zt2}
\ee
We can consider the perturbative expansion of $F_g$ in $s$ with exact dependence on $t$ or the perturbative expansion in $t$ with exact dependence on $s$. In the former case, we find
\be
\aligned
F_2(s,t) & \, = \,  -\frac{1}{240} + \sum_{b\geq 1} \frac{s^b}{b!} C_{2,b}(t^2) ,  \\
F_3(s,t) & \, = \,  \frac{1}{1008} + \sum_{b\geq 1} \frac{s^b}{b!} C_{3,b}(t^2),
\endaligned
\ee
with the constant terms $-1/240$ and $1/ 1008$ reproducing the free energies of the Gaussian model (\ref{FgGauss}) in accordance with the fact that our model reduces to the Gaussian one for $st=0$. 

We find that indeed $C_{2,1}(z)\equiv C_2(z)$ and $C_{3,1}(z)\equiv C_3(z)$ reproduce $C_g(z)$ given in (\ref{Cbb1}), and $C_{2,b}(z)$ and $C_{3,b}(z)$, for $b=2,3$, reproduce generating functions given in (\ref{Cg23}). This confirms that our model correctly reproduces these known results for $b=1,2,3$ backbones in genera two and three. Moreover, we can easily perform expansions to higher order in $s$ revealing generating functions for more backbones. For example for four backbones, we find
\be
\aligned
C_{2,4}(z) & \, = \, \frac{144 z^7}{(1 - 4 z)^{13}} (38675 + 620648 z + 2087808 z^2 \nonumber \\
    & + 1569328 z^3 +  134208 z^4)  , \nonumber  \\
C_{3,4}(z) & \, = \,  \frac{48z^9}{ (1 - 4 z)^{16}} (53416125 + 1194366915 z + 6557325096 z^2 \nonumber \\
    & + 10738411392 z^3 + 4580024832 z^4 + 236239616 z^5),   \nonumber
\endaligned
\ee
and for five backbones, we find
\be
\aligned
C_{2,5}(z) & \, = \, \frac{144 z^8}{(1 - 4 z)^{31/2}} (2543625 + 62424520 z + 375044396 z^2 \nonumber \\
    & + 671666053 z^3 + 314761848 z^4 + 18335696 z^5)  , \nonumber  \\
C_{3,5}(z) & \, = \,  \frac{720z^{10}}{ (1 - 4 z)^{37/2}} (360380790 + 11275076865 z + 95744892585 z^2  \nonumber \\
    & + 282797424880 z^3 + 291167707410 z^4 + 85497242928 z^5 + 3218434848 z^6).  \nonumber
\endaligned
\ee

We can also determine the perturbative expansion in $t$ with exact dependence on $s$. Again scaling variables as in (\ref{T1zt2}) and now expanding up to seven chords, i.e., to order $z^7$, we find
\be
\aligned
F_2& \,  = \,  -\frac{1}{240} + 
 21 s z^4 + (483 s + 1485 \frac{s^2}{2}) z^5 + (6468 s + 26808 s^2 + 
    15015 s^3) z^6 \nonumber \\
      & + (66066 s + 526064 s^2 + 768564 s^3 + 232050 s^4) z^7 + \mathcal{O}(z^8) ,    \nonumber \\
F_3& \,  = \, \frac{1}{1008} + 1485 s z^6 + (56628 s + 225225 \frac{s^2}{2}) z^7 + \mathcal{O}(z^8) ,    \nonumber
\endaligned
\ee
These expansions provide the numbers of configurations with a fixed number of chords and arbitrary numbers of backbones in accordance with Proposition \ref{prop:fin}.  For example, for genus two with seven chords, there are at most four backbones, and for genus three with seven chords, there are at most two backbones. We can easily continue the above expansions to higher powers of $s$ and $t$.



\appendix

\section{Genera zero and one free energies}  \label{preapp}

In this appendix, we determine free energies for genus zero and one, i.e., compute $F_0$ and $F_1$ in general in notation introduced in section \ref{sec:matrix}, and ultimately substitute to derive results for our model in section \ref{sec:solution}.

\subsection{Genus zero} 

The genus zero free energy for the matrix model given in (\ref{Zmatrix}) can be written in a saddle point approximation \cite{BIPZ,Marino05} as
\bea
F_0 & = & -T\int dx \rho(x) V(x) + T^2 \int dx dy \rho(x)\rho(y)\log|x-y| + \nonumber \\
& & + T\lambda \Big(\int \rho(x)dx - t_0  \Big) \nonumber\\ 
& = & -\frac{T}{2} \Big(\int dx \rho(x) V(x) - \lambda t_0 \Big).       \label{F0app}
\eea
In this variational problem, it is crucial to impose the condition preserving the number of eigenvalues, which is fixed by (\ref{t0}). This can be done by introducing and subsequently determining the Lagrange multiplier $\lambda$, where the rest of the notation in the above equation is the same as in section \ref{sec:matrix}. The Lagrange multiplier can be determined from the asymptotic condition on $y(x)$ which (for the one-cut solution) leads to the formula
\be
\lambda = \lim_{\Lambda} \Big( 2 \int_b^{\Lambda} y(x)dx \, -V(\Lambda) - t_0\log \Lambda \Big),
\ee
which is finite as singular terms from the upper integration limit must be canceled by $V(\Lambda)$ and $t_0\log \Lambda$ in the limit of large $\Lambda$. Substituting the results found in section \ref{sec:solution}, we find
\be
\lambda = \frac{t D^2  + S \big( 3 t^2 D^2 + 2 tS (5 - 3  t S) - 4\big) + 
 2t D^2  (3 t S - 1) \log(-D/2)}{8t (t S - 1)},   \label{lambdaOur}
\ee
and $t_0$ defined in (\ref{t0}) takes form
\be
t_0 = \frac{(a - b)^2 \big(3 (a + b) t - 2\big)}{16 T\big( (a + b) t - 2\big)} = \frac{D^2 (3 t S -1)}{4 T (t S - 1)}.  \label{t0Our}
\ee

Moreover, we find the following integral
\be
\int dx \rho(x) V(x)  =   -\frac{1}{64 t^3 T (t S  - 1)}\Big(f_1 + \frac{f_2}{\sqrt{(tS+tD-1)(tS-tD-1)}}  \Big),   \label{rhoV}
\ee
where
\be
\aligned
& f_1  =  16 S - 16 (s + 3 S^2) t + 8 S (-3 D^2 + 2 s + 6 S^2) t^2 +  \\
& \qquad + (D^4 - 16 S^4 + 4 D^2 (2 s + 7 S^2)) t^3 + 3 D^2 S (D^2 - 4 (2 s + S^2)) t^4  \\
& f_2 =  8 (-2 S + 2 (s + 4 S^2) t - 4 S (-D^2 + s + 3 S^2) t^2 +  \\ 
& \qquad +    2 (S^2-D^2) (s + 4 S^2) t^3 +  S ( D^2 (3 s + 4 S^2) -2 D^4 - 2 S^4) t^4).  \nonumber
\endaligned
\ee
Substituting now (\ref{lambdaOur}), (\ref{t0Our}) and (\ref{rhoV}) into (\ref{F0app}), we find $F_0$ expressed in terms of cut endpoints $a$ and $b$ or rather their combinations $S$ and $D$ defined in (\ref{abSD}). We can now substitute expansions of $S$ and $D$ obtained from solving (\ref{ab-equations}) or (\ref{SD-equations}). In particular, substituting perturbative expansions of $S$ and $D$ in $s$ with exact dependence on $t$ as given in appendix \ref{app}, setting $T=1$ and $z=t^2$ as in (\ref{T1zt2}), and neglecting $z$-independent contributions which reproduce the genus zero result in the Gaussian case, we find perfect agreement with known results. 

Indeed, $C_{0,1}(z)\equiv C_0(z)$ reproduces the generating function for the Catalan numbers given in (\ref{C0-Catalan}) while $C_{0,2}(z)$ and $C_{0,3}(z)$ agree with (\ref{Cg23}). Moreover, we can perform expansions in $s$ to arbitrary order and find exact generating functions for arbitrary numbers of backbones. In particular for four and five backbones, we find
\bea
C_{0,4}(z) & = & \frac{24 z^3 (3 + 18 z + 8 z^2)}{(1 - 4 z)^{7}} = \nonumber \\
& = & 72 z^3 + 2448 z^4 + 44544 z^5 + 585984 z^6 + \ldots    \label{C04}  \\
C_{0,5}(z) & = & \frac{24 z^4 (55 + 741 z + 1512 z^2 + 336 z^3)}{(1 - 4 z)^{19/2}} = \nonumber \\
& = & 1320 z^4 + 67944 z^5 + 1765440 z^6 + \ldots     \label{C05}
\eea
To test these results and to provide an idea of their combinatorial complexity,  we show in appendix \ref{app-count} by an explicit enumeration that  in genus zero on four backbones there are indeed 72 diagrams with three chords and 2448 diagrams with four chords.


\subsection{Genus one} As shown in \cite{chekhov2004}, the genus one free energy is given by (\ref{F1def}),  where $M(x)$ is in general defined via (\ref{yMsqrt}) and in our model takes the form (\ref{M-new}), where $a$ and $b$ are given by a solution to the system (\ref{ab-equations}). We can expand this result to arbitrary accuracy in $s$ or $t$ to find generating functions for chord diagrams with arbitrary numbers of backbones or chords. 
The constant term is logarithmic in $T$ and merely reproduces the Gaussian genus one free energy,  
Setting
 $T=1$ in accordance with (\ref{T1zt2}), expanding in $s$, keeping the exact dependence in $t$ and using results presented in appendix \ref{app}, we  find that indeed $C_{1,1}(z)\equiv C_1(z)$ agrees with $C_1(z)$ given in (\ref{Cbb1}) and $C_{1,2}(z)$, $C_{1,3}(z)$ agree with (\ref{Cg23}). This confirms that our approach correctly reproduces known results in genus one for three or fewer backbones. Moreover, performing expansions to higher order in $s$, we can immediately find generating functions for more backbones, for instance,
\bea
C_{1,4}(z) & = & \frac{24 z^5 (715 + 7551z + 12456 z^2 + 2096 z^3) }{(1 - 4 z)^{10}}, \nonumber \\
C_{1,5}(z) & = & \frac{48z^6 (13650 + 253625 z + 975915 z^2 + 840600 z^3 + 86448 z^4)}{(1 - 4 z)^{25/2}}.  \nonumber
\eea

\section{Cut endpoints} \label{app}

It is convenient to present here a solution for cut endpoints $a$ and $b$ obtained from equations (\ref{ab-equations}) in terms of variables introduced in (\ref{abSD})
\be
S = \frac{a+b}{2},\qquad\quad D = \frac{a-b}{2}  .
\ee
Equations (\ref{SD-equations}) provide a solution which is perturbative in $s$ and exact in $t$: 
\be
\aligned
D & = -2 T^{1/2} - \frac{2 s t^2 T^{1/2}}{(1 - 4 t^2 T)^{3/2}} - \frac{3 s^2 t^4 T^{1/2} (3 + 4 t^2 T)}{(1 - 4 t^2 T)^4} \nonumber \\
  & - \frac{ s^3 t^6 T^{1/2} (47 + 308 t^2 T + 128 t^4 T^2)}{(1 - 4 t^2 T)^{13/2}} \nonumber \\
  & - \frac{ s^4 t^8 T^{1/2} (1051 + 15756 t^2 T + 33168 t^4 T^2 + 6720 t^6 T^3)}{4 (1 - 4 t^2 T)^9} \nonumber \\
  & - \frac{3 s^5 t^{10} T^{1/2} (2037 + 53448 t^2 T + 262224 t^4 T^2 + 
    272640 t^6 T^3 + 32768 t^8 T^4)}{4 (1 - 4 t^2 T)^{23/2} }
    + \mathcal{O}(s^6) \nonumber \\
S & = \frac{s t}{(1 - 4 t^2 T)^{3/2}} + \frac{ 2 s^2 t^3 (1 + 8 t^2 T)}{(1 - 4 t^2 T)^4}  \nonumber \\
  & + \frac{ s^3 t^5 (7 + 166 t^2 T + 280 t^4 T^2)}{(1 - 4 t^2 T)^{13/2}} \nonumber \\
  & + \frac{2 s^4 t^7 (15 + 684 t^2 T + 3552 t^4 T^2 + 2560 t^6 T^3)}{(1 - 4 t^2 T)^9} \nonumber \\
  & + \frac{s^5 t^9 (143 + 10408 t^2 T + 107094 t^4 T^2 + 240144 t^6 T^3 + 
   96096 t^8 T^4)}{(1 - 4 t^2 T)^{23/2} }
  + \mathcal{O}(s^6) \nonumber
\endaligned
\ee

                 
\section{Explicit diagram counting}  \label{app-count}

To illustrate the combinatorial complexity of chord diagrams and fully appreciate the generating functions which we have computed, we directly count in this appendix  the number of genus zero diagrams on four backbones with three or four chords. In this way, we also confirm correctness of the coefficients
$c_{0,4}(3)=72$ and $c_{0,4}(4)=2448$ in the genus zero expansion on four backbones in (\ref{C0414intro}). 

Our first task is to distribute endpoints of chords along backbones in all possible inequivalent ways.  If there are $n$ chords,
then there are $2n$ such endpoints, perhaps better apprehended as half-chords.  We shall denote a partition of
$\{1,2,\cdots 2n\}$ into $n_i$ sets of cardinality $i\geq 1$ as $\{ n_1,\ldots ,n_K\}$ where $\sum _{k=1}^K n_k=2n$.

We must then count how many ways there are to produce a connected chord diagram of the correct genus by pairing up these
putative
endpoints.  In fact in the cases under consideration, the genus is always zero by Proposition \ref{prop:fin}, so here we must  count
simply the number of connected such pairings.

\bigskip

To begin, consider diagrams on four backbones  with three chords.  The partitions of six
are:  $\{3,1,1,1\}$ and $\{2,2,1,1\}$. 
\begin{itemize}
\item For $\{3,1,1,1\}$,
all possible linearly ordered sets of backbones arise by embedding them in the real line with the induced linear
order and orientations.  
To achieve such an embedding, we merely determine which among the four
backbones has three half-chords.
Order these half-chords 1, 2, 3, which 
connect uniquely to the other 
backbones occurring in the order 1, 2, 3
along the real axis.  There are thus 4 times 3! = 24 such chord diagrams.

\smallskip
\item For $\{2,2,1,1\}$, there are ${4 \choose 2} = 6$ places to put the two 2's.
By connectivity, there must be a chord between the two 2's,
and there are four ways to do this.  The two remaining half-chords on the 2's
again come in an order along the real axis as do the two remaining
1's, giving another factor 2!, for a total of $6\cdot 4\cdot 2!=48$.

\end{itemize}
The grand total is thus $24+48=72$ diagrams in agreement with the first term in the expansion of the generating function $C_{0,4}$ given in  (\ref{C0414intro}).

\bigskip

We next consider the more involved case of four backbones with four chords. There are the five partitions $\{5,1,1,1\}$, $\{4,2,1,1\}$, $\{3,2,2,1\}$, $\{2,2,2,2\}$, and $\{3,3,1,1\}$ of the eight half-chords into four backbones, and:
\begin{itemize}

\item For $\{5,1,1,1\}$, there are four locations for the 5. Two half-chords on 5 must
be connected, and there are
${5 \choose2}=10$ of these.  The remaining half-chords on 5 connect
arbitrarily to the 1's giving another
factor 3!, for a total of $4\cdot 10\cdot 6=240$ diagrams for $\{5,1,1,1\}$.

\smallskip

\item For $\{4,2,1,1\}$ there are four locations for 4 and three for 2.
There must be at least one 4-2 chord by connectivity and at most two
by cardinality.

If there is a unique chord 4-2, among the $4\cdot 2$ such possibilities, then the
other half-chord
on the 2 must connect to a 1 by definition, and there are two of these;
the remaining 1 then  can connect to any of the three remaining half-chords;
there are a total of $4\cdot 3\cdot 4\cdot 2\cdot 2\cdot 3=576$ in this sub-case.

There are ${4\choose 2}=6$ times two possible pairs of 4-2 chords and two
ways to
produce a connected diagram connecting the remaining 1's to the remaining
half-chords on 4; there are  a total of $4\cdot 3 \cdot 6\cdot 2\cdot 2=288$ in this sub-case.

\smallskip\noindent There are  a total of $576+288=864$ chord diagrams for $\{4,2,1,1\}$.\smallskip

\item For $\{3,2,2,1\}$, there are four locations for 3 and three for 1. There must be at
least one 3-2 chord by connectivity.

If there is exactly one 3-2 chord among the $3\cdot 4$ possibilities, then
the other half-chord from the 2 must go to the
other 2, with two possibilities.  This gives $4\cdot 3\cdot 4\cdot 3\cdot 2 = 288$
in this sub-case.

If there are exactly two 3-2 chords, then their 2-endpoints must lie
on different backbones by connectivity.  There are thus ${3\choose 2}=3$
times eight such pairs
of chords and a unique completion to a diagram with exactly two 3-2 chords.
There are a total of $4\cdot 3\cdot 3\cdot 8=288$ in this sub-case.

There are $4\cdot 3\cdot 2$ configurations of three 3-2 chords with a unique
completion, so a total
of $4\cdot 3\cdot 4\cdot 3\cdot 2=288$ in the sub-case of exactly three 3-2 chords.

\smallskip\noindent There are a total of $288+288+288=864$ diagrams for $\{3,2,2,1\}$.\smallskip

\item For $\{2,2,2,2\}$, there can be no two chords with the same backbone endpoints by
connectivity,
so the first chord on the first backbone lands in one of six places and
the second in one of four places on different
backbones.   There are two connected completions of any such partial diagram, for 
a total of $6\cdot 4\cdot 2=48$ diagrams for $\{2,2,2,2\}$.

\smallskip

\item For $\{3,3,1,1\}$, there are ${4\choose 2}$ locations for the 3's.  A fixed 1 must
connect to a 3 in one of three possible ways
times two for the two 3's.
Either this 3 connects to the other 1 in two possible ways and to the
other 3 in three possible ways, or not, in which case, the other 1
connects to the other 3 in three  possible ways with two possible
connected completions for a total of $6\cdot 6\cdot (2\cdot 3+3\cdot 2) = 432$ diagrams for $\{3,3,1,1\}$.

\end{itemize}
Overall, we find
$
240+864+864+48+432=2448
$
chord diagrams in agreement with the second term $c_{0,4}(4)$ in the expansion of the generating function $C_{0,4}$ given in  (\ref{C0414intro}).


\bigskip

\noindent \thanks{Acknowledgements: JEA and RCP are supported by the Centre 
for Quantum Geometry of Moduli Spaces which is funded by the Danish National Research Foundation.
The research of LCh is supported by the Russian Foundation for Basic
Research (Grants Nos 10-02-01315-a, 11-01-12037-ofi-m-2011, and
11-02-90453-Ukr-f-a) by the Program Mathematical Methods of Nonlinear
Dynamics and by the Grant of Supporting Leading Scientific Schools
NSh-4612.2012.1.
The research of PS is supported by the DOE grant DE-FG03-92-ER40701FG-02,
the European Commission under the Marie-Curie International Outgoing Fellowship Programme, 
and the Foundation for Polish Science.  RCP also acknowledges the kind support of
Institut Henri Poincar\'e where parts of this manuscript were written.
}


\end{document}